\documentclass[12pt]{article}%
\usepackage{graphicx}%
\usepackage[normalem]{ulem}
\usepackage{amsmath,amsthm,amsfonts,
amsbsy,amssymb,upref,enumerate,bigstrut,
color,mathtools,mathrsfs,float,bm,dsfont,scalefnt,mathtools}
\usepackage{natbib}
\usepackage{lipsum}
\usepackage{lmodern}
\usepackage{stmaryrd}
\usepackage{authblk} 
\usepackage{subfigure}
\usepackage{pgfplots}
\usepackage{caption}
\usepackage{tikz}
\usetikzlibrary{intersections}

\usepackage[colorlinks=true,linkcolor=blue,citecolor=blue]{hyperref}

\usepackage[left=1in,top=1.2in,right=0.7in]{geometry}

\theoremstyle{definition}

\newtheorem{theorem}{Theorem}[section]

\newtheorem{proposition}[theorem]{Proposition}
\newtheorem{remark}[theorem]{Remark}
\numberwithin{equation}{section} 

\makeatletter
\def\@seccntformat#1{\@ifundefined{#1@cntformat}%
	{\csname the#1\endcsname\quad}
	{\csname #1@cntformat\endcsname}
}
\makeatother

\newif\ifShowComments
\ShowCommentstrue
\def\strutdepth{\dp\strutbox}
\def\druk#1{\strut\vadjust{\kern-\strutdepth
        {\vtop to \strutdepth{%
                \baselineskip\strutdepth\vss
                        \llap{\hbox{#1}\quad}\null}}}}

\title{\bf
Unit-log-symmetric models: Characterization, statistical properties and its use in analyzing internet access data
}

\author[1, 2]{Roberto Vila \thanks{rovig161@gmail.com}}
\author[2]{Narayanaswamy Balakrishnan,  \thanks{bala@mcmaster.ca
} }
\author[1]{Helton Saulo  \thanks{heltonsaulo@gmail.com} }
\author[1]{Peter Zörnig \thanks{peter@unb.br}}
\affil[1]{Department of Statistics, University of
	 Bras\'ilia, Bras\'ilia, Brazil}
\affil[2]{
Department of Mathematics and Statistics, McMaster University, Hamilton, Ontario, Canada}
\begin{document}
\maketitle

\begin{abstract}
{
We present here a unit-log-symmetric model based on the bivariate log-symmetric distribution.
It is a flexible family of distributions over the interval $(0, 1)$. We then discuss its mathematical properties such as stochastic representation, symmetry, modality, moments,  quantile function, entropy and maximum likelihood estimators, paying particular attention to the special cases of unit-log-normal,  unit-log-Student-$t$ and unit-log-Laplace distributions. Finally, some empirical results and practical illustrations are presented.
}
\end{abstract}


\section{Introduction}
\noindent
Distribution theory has mostly focused on unbounded distributions. However, in the recent years, bounded distributions have also become a subject of interest. In many practical applications, such as in medicine, biology, economics and finance, bounded data arise naturally. Some of the most common models used in this case are the beta distribution and its generalizations, Kumaraswamy and Topp-Leone distribution. 
In order to have including flexible models for real world problems, it is often desirable to have a whole family of distributions available. For this purpose, transformations have been developed in the statistical literature which construct a bounded distribution for any given distribution over the real line $(-\infty, \infty)$ or over $(0, \infty)$, respectively, such as logistic transformation.

In this work, we study a transformation that contains the latter as a special case (for $\mu=0$ and $\sigma=1$). However, our transformation is derived from the bivariate log-symmetric distribution
which is defined for the first time in this paper.
We introduce  the unit-log-symmetric (ULS) distribution in Section \ref{Sec:2}. We pay special attention to the cases of unit-log-normal, unit-log-Student-$t$ and unit-log-Laplace distributions. In Section \ref{Sec:3}, mathematical properties of these models are studied. Stochastic representations of the ULS are given which are used to simplify the density considerably. The symmetry of the ULS density is established and then the above-mentioned special cases are explicitly formulated, and graphically illustrated. The modality is also discussed in detail. 
The unit-log-normal density may be unimodal or bimodal; the unit-log-Student-$t$ density may have a bathtub shape or possess two minimum points; the unit-log-Laplace density may be unimodal or possess two minimum points. The quantile function is studied in Subsection \ref{Quantile function} while the moments 
and the Shannon entropy in Subsections \ref{Moments} and \ref{Shannon entropy}, respectively. Finally, the maximum likelihood (ML) estimators are studied in Subsection \ref{ml_simulation_res}. In Section \ref{Sec:5}, the Monte Carlo simulation results are presented for demonstrating the performance of the estimators and for assessing the empirical distribution of the residuals. The paper ends with the application of the unit-log-normal, unit-log-Student-$t$ and unit-log-Laplace distributions for modeling an internet access data. Specifically, the three models are fitted to the proportion of 
world population using Internet (Section \ref{Sec:6}). Finally, in Section \ref{Sec:7}, some concluding remarks are made.


\section{Unit-Log-Symmetric Distribution }\label{Sec:2}
\noindent

We say that a continuous random variable $W$, with support 
$(0, 1)$, follows a unit-log-symmetric  (ULS) distribution with parameter $\sigma_\rho=\sigma\sqrt{2(1-\rho)}$, denoted by $W\sim {\rm ULS}(\sigma_\rho,g_c)$, if its probability density function (PDF) is given by
\begin{eqnarray}\label{ULS-PDF}
f_W(w;\sigma_\rho)
=
{1\over w(1-w) \sigma^2\sqrt{1-\rho^2} Z_{g_c}}\, 
\int_{0}^{\infty} 
{1\over t}\,
g_c\Biggl(
{\widetilde{t}_w^{\,2}-2\rho\widetilde{t}_w\widetilde{t}+\widetilde{t}^{\,2}
	\over 
	1-\rho^2}
\Biggr)  
{\rm d}t,  
\quad 
0<w<1,
\\[0,3cm]
\text{where} \quad 
t_w= \bigg({w\over 1-w}\bigg) t, 
\
\widetilde{t}_w
=
\log\biggl[\Bigl({t_w\over \eta}\Bigr)^{1/\sigma}\biggr], 
\
\widetilde{t}
=
\log\biggl[\Bigl({t\over \eta}\Bigr)^{1/\sigma}\biggr],
\ \eta=\exp(\mu),
\nonumber
\end{eqnarray}
with $\mu\in\mathbb{R}$, $\sigma>0$ and $\rho\in(-1,1)$. Moreover, $Z_{g_c}$ is a positive constant defined as
\begin{align}\label{partition function}
Z_{g_c}
&=
\int_{0}^{\infty}\int_{0}^{\infty}
{1\over t_1t_2\sigma^2\sqrt{1-\rho^2}}\,
g_c\Biggl(
{\widetilde{t}_1^{\,2}-2\rho\widetilde{t}_1\widetilde{t}_2+\widetilde{t}_2^{\,2}
	\over 
	1-\rho^2}
\Biggr)\, {\rm d}t_1{\rm d}t_2
=
\pi \int_{0}^{\infty}  g_c(u)\,{\rm d}u,
\end{align}
with $g_c$ being a scalar function, referred to as the density generator, which may or may not depend on extra parameters considered known
\citep{Fang1990}.

\begin{remark}
We have defined the PDF in \eqref{ULS-PDF} as a uniparameter function for two reasons: first, to avoid identification problems (see Proposition \ref{Proposition-identifiable}) and second, because the shape of this PDF depends only on the parameter $\sigma_\rho$ (see Subsection \ref{Modality}).
\end{remark}



For simplicity in presentation, we focus our attention specifically on the density generators of unit-log-normal, unit-log-Student-$t$ and unit-log-Laplace (Table \ref{table:1}), though other forms of density generators can be handled in an analogous manner.
    \begin{table}[H]
    	\caption{Normalization constants $(Z_{g_c})$ and density generators $(g_c)$ for some ULS distributions.}
    	\vspace*{0.15cm}
    	\centering
    	\begin{tabular}{llll}
    		\hline
    		Distribution
    		& $Z_{g_c}$ & $g_c$ & Parameter
    		\\ [0.5ex]
    		\noalign{\hrule height 1.0pt}
    		Unit-log-normal
    		& $2\pi$ & $\exp(-x/2)$ & $-$
    		\\ [1ex]
    		Unit-log-Student-$t$
    		& ${{\Gamma({\nu/ 2})}\nu\pi\over{\Gamma({(\nu+2)/ 2})}}$
    		& $(1+{x\over\nu})^{-(\nu+2)/ 2}$  &  $\nu>0$
    		\\ [1ex]
      		Unit-log-Laplace
      		& $\pi$  & $K_0(\sqrt{2x})$ & $-$
      		\\ [1ex]
    	\hline
    	\end{tabular}
    	\label{table:1} 
    \end{table}
Here, in Table \ref{table:1}, 
$K_0(u)=\int_0^\infty t^{-1} \exp(-t-{u^2\over 4t}) \,{\rm d}t/2$, $u>0$, is the Bessel function of the third kind
\cite[for more details on the main properties of $K_0$, one may refer to appendix of][]{Kotz2001}.
\section{Some basic properties} \label{Sec:3}
\noindent

We now discuss some mathematical properties of the unit-log-symmetric distribution stated in the last section.
For this purpose, the following definition of bivariate log-symmetric  distribution
becomes essential.

We say that a continuous random vector $(T_1,T_2)$, with $T_1$ and $T_2$ being identically distributed, follows a bivariate log-symmetric (BLS) distribution  if its joint PDF is given by 
\begin{eqnarray*}
f_{T_1,T_2}(t_1,t_2;\boldsymbol{\theta})
=
{1\over t_1t_2\sigma^2\sqrt{1-\rho^2}Z_{g_c}}\,
g_c\Biggl(
{\widetilde{t}_1^{\,2}-2\rho\widetilde{t}_1\widetilde{t}_2+\widetilde{t}_2^{\,2}
	\over 
	1-\rho^2}
\Biggr),
\quad 
t_1,t_2>0,
\\[0,3cm]
\text{where}\quad 
\widetilde{t_i}
=
\log\biggl[\Bigl({t_i\over \eta}\Bigr)^{1/\sigma}\biggr], \ \eta=\exp(\mu), \ i=1,2, \nonumber
\end{eqnarray*}
with $\boldsymbol{\theta}=(\eta,\sigma,\rho)$ being the parameter vector, $\mu\in\mathbb{R}$, $\sigma>0$, $\rho\in(-1,1)$, 
and $g_c$ and $Z_{g_c}$ are as given in \eqref{ULS-PDF} and \eqref{partition function}, respectively. 
Observe that $Z_{g_c}$ in \eqref{partition function} is
the normalization constant for $f_{T_1,T_2}$ and that the joint PDF of $(T_1, T_2)$ can be obtained as the PDF of an  exponential transformation of an elliptically symmetric vector \cite[Chapter 13, p. 591]{Bala2009}.

\subsection{Stochastic representation}
A simple observation shows that if $(T_1,T_2)\sim {\rm BLS}(\boldsymbol{\theta},g_c)$, then
\begin{align}\label{CDF-general}
	\mathbb{P}\bigg({T_1\over T_1+T_2}\leqslant w\bigg)=\mathbb{P}\bigg({T_1\over T_2}\leqslant {w\over 1-w}\bigg),  
	\quad 
	0<w<1,
\end{align}
and so,
\begin{align*}
	f_{{T_1\over T_1+T_2}}(w)
	=
	{1\over (1-w)^2}\, f_{\frac{T_1}{T_2}}\bigg({w\over 1-w}\bigg)
	=
	{1\over (1-w)^2}\, 
	\int_{0}^{\infty} 
	t f_{T_1, T_2}(t_w, t;\boldsymbol{\theta})\,  
	{\rm d}t
	=
	f_W(w;\sigma_\rho).
\end{align*}
This means that if $(T_1,T_2)\sim {\rm BLS}(\boldsymbol{\theta},g_c)$,
 the ULS random variable $W$ admits the stochastic representation
\begin{align}\label{stochastic representation}
	W={T_1\over T_1+T_2}.
\end{align}


\subsection{Characterizations}\label{Characterization of the ULS distribution}
%
From the classic stochastic representation of a bivariate
elliptical vector \cite[Subsection 13.2.3, p. 593]{Bala2009}, it follows that $(T_1,T_2)\sim {\rm BLS}(\boldsymbol{\theta},g_c)$ has the following representation:
\begin{align} \label{rep-stoch-biv-student-t}
\begin{array}{lllll}
&T_1=\eta \exp(\sigma Z_1),
\\[0,3cm]
&T_2=\eta 
\exp(\sigma[ {\rho} Z_1+\sqrt{1-\rho^2} Z_2]),
\end{array}
\end{align} 
where $Z_1=RDU_1$ and $Z_2=R\sqrt{1-D^2}U_2$ with
$U_1$, $U_2$,$R$ and $D$ being mutually independent, $\rho\in(-1,1)$, $\eta=\exp(\mu)$, $\mathbb{P}(U_i = -1) = \mathbb{P}(U_i = 1) = 1/2$, $i=1,2$, and the variables $D$ and $R$ have PDFs $f_D(d)={2/(\pi\sqrt{1-d^2})}, \ d\in(0,1)$, and 
$	f_R(r)={2r g_c(r^2)/[\int_{0}^{\infty}
	g_c(u) \, {\rm d}{u}]}, \ r>0$, respectively. 

Therefore,
$
{T_1/ T_2}
=
\exp(\sigma[(1 -{\rho}) Z_1-\sqrt{1-\rho^2} Z_2]).
$
and so we obtain
\begin{align*}
	\mathbb{P}\bigg({T_1\over T_2}\leqslant {w\over 1-w}\bigg)
	=
	\mathbb{P}(
Z_\rho
	\leqslant 
A(w)
),
\end{align*}
where 
\begin{align}\label{sigma-and-Z-rho}
Z_\rho={1\over \sqrt{2(1-\rho)}}\,[(1-{\rho}) Z_1-\sqrt{1-\rho^2} Z_2]
\end{align}
and
\begin{align}\label{def-A}
A(w)=\log\bigg[\Big({w\over 1-w}\Big)^{1/\sigma_\rho}\bigg].
\end{align}

Thus, from \eqref{CDF-general} and \eqref{stochastic representation}, the cumulative distribution function (CDF) of  $W\sim {\rm ULS}(\sigma_\rho,g_c)$,  denoted by $F_W(w;\sigma_\rho)$, can be expressed as
\begin{align}\label{general cdf formula}
F_W(w;\sigma_\rho)
=
	\mathbb{P}(
Z_\rho
\leqslant 
A(w)
),  
\quad 
0<w<1.
\end{align}

From \eqref{general cdf formula}, the following result follows immediately.
\begin{proposition}\label{Proposition-identifiable}
	If the distribution of $Z_\rho$ depends (or not) only on the parameters of the density generator $g_c$, then the ULS distribution is identifiable.
\end{proposition}

The inverse function of $A$, denoted by $A^{-1}$, is given by
\begin{align}\label{inverse-A}
	A^{-1}(w)={1\over 1+\exp(-\sigma_\rho w)}.
\end{align}
So, using this notation, the following result follows immediately from \eqref{general cdf formula}.
\begin{proposition}[Another stochastic representation]\label{Another stochastic representation}
Given the distribution of $Z_\rho$, we have $A^{-1}(Z_\rho)\sim {\rm ULS}(\sigma_\rho,g_c)$.  On the other hand, if $W\sim {\rm ULS}(\sigma_\rho,g_c)$, then $A(W)$ and $Z_\rho$ have the same distribution.
\end{proposition}

Differentiating \eqref{general cdf formula} with respect to $w$, the ULS PDF in \eqref{ULS-PDF} of $W$ is characterized as
\begin{align}\label{general pdf formula}
f_W(w;\sigma_\rho)
=
{1\over w(1-w) \sigma_\rho}\, 
f_{Z_\rho}(A(w)),  
\quad 
0<w<1,
\end{align}
where $f_{Z_\rho}$ denotes the PDF of $Z_\rho$.

This means that, the distribution of $W$ is completely determined by the distribution of $Z_\rho$.

\begin{proposition}[Symmetry]\label{point of symmetry}
	The ULS PDF in \eqref{ULS-PDF} is symmetric around $w_0 = 1/2$ provided that the PDF of  $Z_\rho$ is symmetric around $z_0=0$.
	Furthermore, in this case, the median and the mean of a ULS distribution both occur at $w_0 = 1/2$.
\end{proposition}
\begin{proof}
	A simple algebraic manipulation shows that
	$$
	f_W\left({1\over 2}-w;\sigma_\rho\right)
	=
	f_W\left({1\over2}+w;\sigma_\rho\right),
	\quad
	\forall 0<w<1,
	$$
	provided that $f_{Z_\rho}(-z)=f_{Z_\rho}(z)$, $\forall z\in\mathbb{R}$.
	The required result then follows.
\end{proof}

\begin{remark}
	Since the PDFs of ${Z_1}$ and ${Z_2}$ are even functions \cite[see Proposition 1; Item ii, of][]{Saulo2022},
	and provided that $Z_1$ and $Z_2$ are independent, we have that $f_{Z_\rho}(-z)=f_{Z_\rho}(z)$, $\forall z\in\mathbb{R}$.
\end{remark}

\subsubsection{Unit-log-normal}
\label{Unit-log-normal distribution}
It is well-known that the bivariate log-normal distribution admits a stochastic representation as in \eqref{rep-stoch-biv-student-t},  where $Z_1$ and $Z_2$ are 
independent and identically distributed (i.i.d.)  standard normal  random variables \cite[Corollary 1 of][]{Saulo2022}. Consequently, $Z_\rho$ in \eqref{sigma-and-Z-rho} is standard normally distributed. So, by using \eqref{general cdf formula} and \eqref{general pdf formula}, we obtain
\begin{align}\label{ULS-PDF-normal}
F_W(w;\sigma_\rho)
=
\Phi(A(w)), 
\quad
f_W(w;\sigma_\rho)
=
{1\over w(1-w) \sigma_\rho}\, 
\phi(A(w)), 
\quad 
0<w<1,
\end{align}
where $\phi$ and $\Phi$ denote the PDF and CDF of a standard normal distribution, respectively.

\subsubsection{Unit-log-Student-$t$}
\label{Unit-log-Student-t distribution}
It is well-known that the bivariate log-Student-$t$ distribution has a stochastic representation as in \eqref{rep-stoch-biv-student-t}, where $Z_1=Z_1^* \sqrt{\nu/Q}\sim t_\nu$ and $Z_2= Z_2^* \sqrt{\nu/Q}\sim t_\nu$ \cite[Corollary 2 of][]{Saulo2022}. Here, $Q\sim\chi^2_\nu$ (chi-square with $\nu$ degrees of freedom) is independent by of $Z_1^*$ and ${\rho} Z_1^* +\sqrt{1-\rho^2} Z_2^*$, 
whereas $Z_1^*$ and $Z_2^*$ are i.i.d.  standard normal  random variables.

As
$[1/\sqrt{2(1-\rho)}][(1-{\rho}) Z_1^*-\sqrt{1-\rho^2} Z_2^*]\sim 
N(0,1)$,
we have
\begin{align}\label{distributed-t}
Z_\rho
=
{1\over \sqrt{2(1-\rho)}}\,[(1-{\rho}) Z_1^*-\sqrt{1-\rho^2} Z_2^*]
\sqrt{\nu\over Q}\sim t_\nu.
\end{align}
By combining \eqref{general cdf formula} and \eqref{general pdf formula} with \eqref{distributed-t}, we obtain 
\begin{align*}
F_W(w;\sigma_\rho)
=
F_\nu(A(w)), 
\quad
f_W(w;\sigma_\rho)
=
{1\over w(1-w)\sigma_\rho}\, 
f_\nu(A(w)), 
\quad 
0<w<1,
\end{align*}
where $f_\nu$ and $F_\nu$ denote the PDF and CDF of a Student-$t$ distribution with $\nu$ degrees of freedom, respectively.

\subsubsection{Unit-log-Laplace}\label{Unit-log-Laplace}
By using the general algorithm for simulation of symmetric
bivariate Laplace variables \cite[Subsection 5.1.4, p. 234]{Kotz2001}, it follows that the bivariate log-Laplace distribution has a stochastic representation as in \eqref{rep-stoch-biv-student-t}, where $Z_1=Z_1^* \sqrt{B}\sim {\rm Laplace}(0,1/\sqrt{2})$ and $Z_2= Z_2^* \sqrt{B}\sim {\rm Laplace}(0,1/\sqrt{2})$. Here, $B$ is a standard exponential variable independent by of $Z_1^*$ and ${\rho} Z_1^* +\sqrt{1-\rho^2} Z_2^*$, 
whereas $Z_1^*$ and $Z_2^*$ are i.i.d.  standard normal  random variables.

As
$[1/\sqrt{2(1-\rho)}][(1-{\rho}) Z_1^*-\sqrt{1-\rho^2} Z_2^*]\sim 
N(0,1)$,
we have
\begin{align}\label{distributed-Laplace}
Z_\rho
=
{1\over \sqrt{2(1-\rho)}}\,[(1-{\rho}) Z_1^*-\sqrt{1-\rho^2} Z_2^*]
\sqrt{B}\sim {\rm Laplace}(0,1/\sqrt{2}).
\end{align}
Then, by \eqref{general cdf formula}, \eqref{general pdf formula} and \eqref{distributed-t} we obtain  
\begin{align}\label{ULS-PDF-Laplace}
F_W(w;\sigma_\rho)
=
F_\ell(A(w)), 
\quad
f_W(w;\sigma_\rho)
=
{1\over w(1-w)\sigma_\rho}\, 
f_\ell(A(w)), 
\quad 
0<w<1,
\end{align}
where $F_\ell$ and $f_\ell$ denote the CDF and PDF of Laplace distribution with location parameter $0$ and scale parameter $1/\sqrt{2}$, respectively.

Figure \ref{figpdf:1} displays different shapes of the unit-log-symmetric PDFs for different choices of parameters. From this figure, we observe that the parameter $\sigma_{\rho}$ controls the shape of the unit-log-normal, unit-log-Laplace and unit-log-Student-$t$ densities. We also note that the PDFs of these three models are symmetric. Finally, from Figure \ref{figpdf:1}(d), we see the clear effect the kurtosis parameter $\nu$ has on the unit-log-Student-$t$ density.

\begin{figure}[h!]
\vspace{-0.25cm}
\centering
\subfigure[Unit-log-normal]{\includegraphics[height=6.5cm,width=6.5cm]{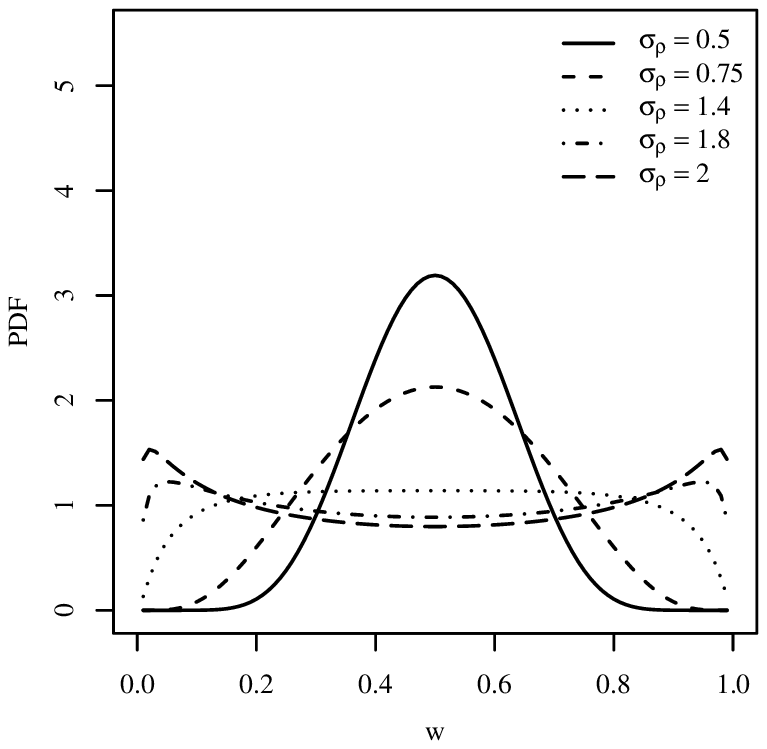}}
\subfigure[Unit-log-Laplace]{\includegraphics[height=6.5cm,width=6.5cm]{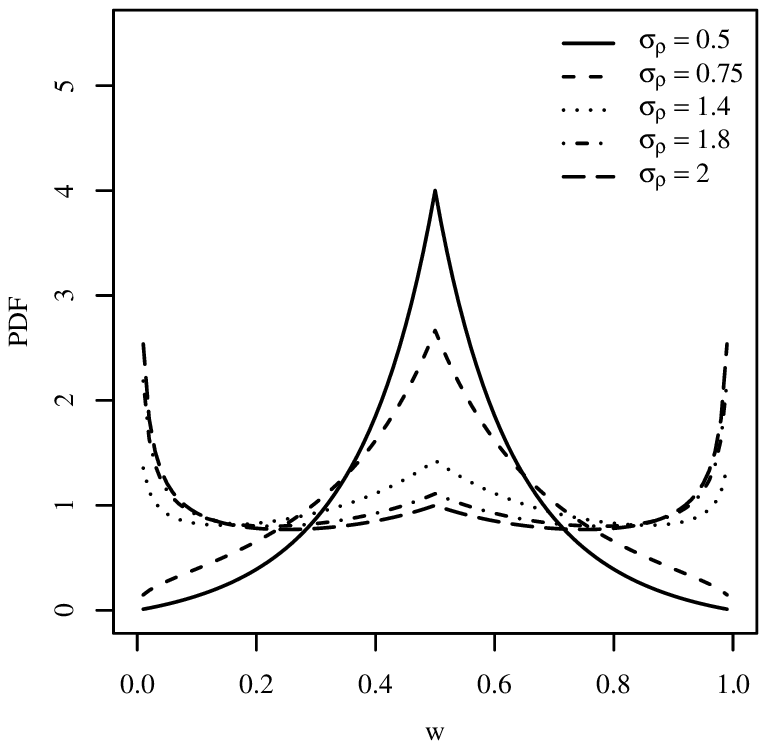}}
\subfigure[Unit-log-Student-$t$]{\includegraphics[height=6.5cm,width=6.5cm]{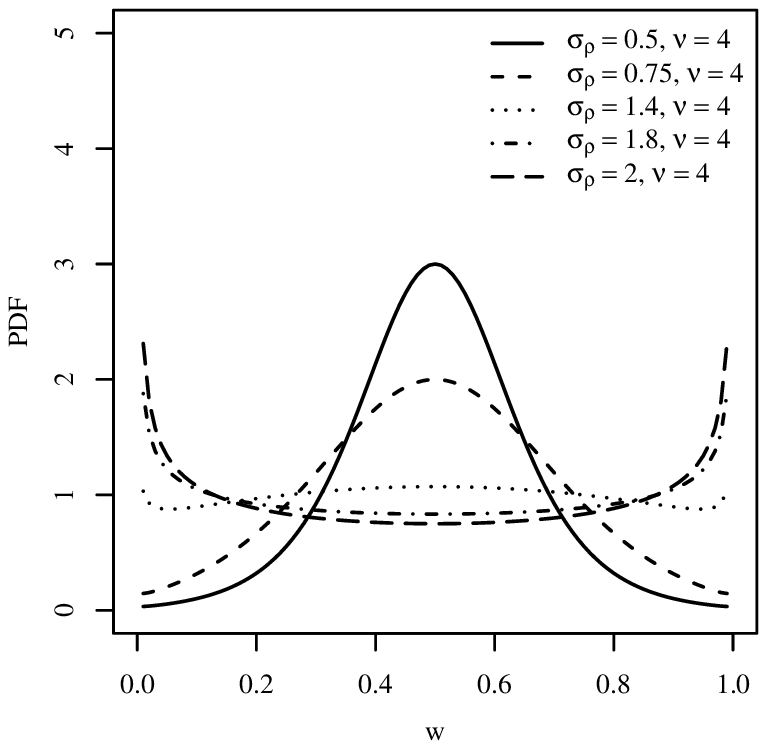}}
\subfigure[Unit-log-Student-$t$]{\includegraphics[height=6.5cm,width=6.5cm]{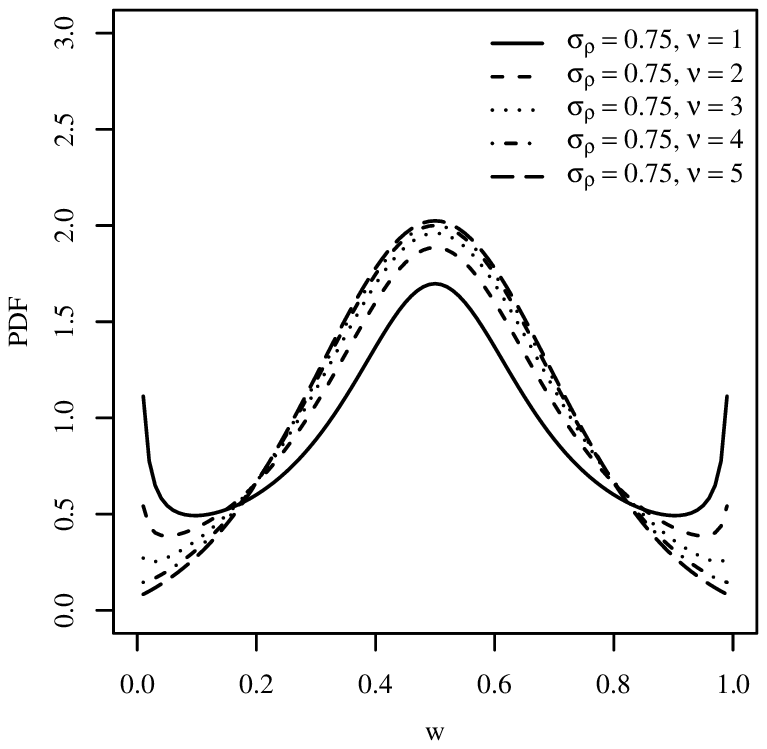}}
\vspace{-0.2cm}
\caption{{Unit-log-symmetric PDFs for some choices of parameters.}}\label{figpdf:1}
\end{figure}

\subsection{Modality}\label{Modality}

Differentiating $A$ in \eqref{def-A} with respect to $w$, we get
\begin{align}\label{a-derivatives}
	A'(w)={1\over w(1-w)\sigma_\rho}
	\quad \text{and}\quad 
	A''(w)=(2w-1)[A'(w)]^2 \sigma_\rho.
\end{align}

Suppose the derivative of $f_W(w;\sigma_\rho)$ with respect to $w$  exists at every point in its domain. In the case there is a countable number of points on its domain where $f_W$ is not differentiable, by definition, these are considered to be critical points of $f_W$.
Thus, from \eqref{general pdf formula},
we have $f_W(w;\sigma_\rho)=A'(w) f_{Z_\rho}(A(w))$ so that
\begin{align*}
	f'_W(w;\sigma_\rho)
	=
	A''(w) f_{Z_\rho}(A(w))+[A'(w)]^2f'_{Z_\rho}(A(w)).
\end{align*} 
If $f'_{Z_\rho}(z)=-r(z)f_{Z_\rho}(z)$, for some real-valued function $r$,  then
\begin{align*}
f'_W(w;\sigma_\rho)
=
\{A''(w)-[A'(w)]^2 r(A(w))\} f_{Z_\rho}(A(w))
\stackrel{\eqref{a-derivatives}}{=}
\{(2w-1)\sigma_\rho-r(A(w))\} A'(w) f_W(w;\sigma_\rho).
\end{align*}
Hence,
\begin{align*}
f'_W(w;\sigma_\rho)=0  \quad
\Longleftrightarrow\quad
(2w-1)\sigma_\rho=r(A(w)) \quad
\Longleftrightarrow\quad
\left[2 A^{-1}(A(w))-1\right]\sigma_\rho
=r(A(w)).
\end{align*} 
Therefore, by using \eqref{inverse-A}, a critical point $w$ of the ULS PDF \eqref{ULS-PDF} is such that
\begin{align*}
\left[{2\over 1+\exp(-\sigma_\rho A(w))}-1\right]\sigma_\rho
=r(A(w)), 	
\end{align*} 
or, equivalently,
\begin{align}\label{eq-critical point}
\tanh\biggl({\sigma_\rho^2 \over 2}\, y\biggr)
=
{r(\sigma_\rho y)\over \sigma_\rho},
\quad \text{with} \ y={A(w)/ \sigma_\rho}.
\end{align}

\subsubsection{Unit-log-normal}
In this case, from Subsubsection \ref{Unit-log-normal distribution}, we have $Z_\rho\sim N(0,1)$, and then $r(z)=z$. So, from \eqref{eq-critical point}, the equation for the critical points of the ULS PDF is given by	
\begin{align*}
	\tanh\biggl({\sigma_\rho^2 \over 2}\, y\biggr)
	=
	y, \quad \text{with} \ y={A(w)/ \sigma_\rho},
\end{align*}
if and only if 
\begin{align}\label{tanh-equation}
{\sigma_\rho^2 \over 2}\, y={{\rm arctanh}(y)}.
\end{align}
%
%
%
\pgfkeys{
	/pgfplots/Axis Style/.style={
		width=9.9cm, height=6cm, 
		axis x line=middle,
		axis y line=middle,
		yticklabel style={inner ysep=0pt, anchor=east},
		ytick={0,0,0},
		samples=100,
		ymin=-2.0, ymax=2.0,
		xmin=-1.25, xmax=1.25,
		domain=-1.25:1.25
	}
}
\begin{figure}[!h]
\centering
\begin{minipage}{0.5\textwidth}
	\centering
\pgfkeys{/pgf/declare function={arctanh(\x) = 0.5*(ln(1+\x)-ln(1-\x));}}
\begin{tikzpicture}
\begin{axis}[
Axis Style,
xtick={
	 -1.0 ,0.0, 1.0
},
xlabel= {$y$},
	xticklabels={
		$-1$, $0$, $1$,
	},
]
\addplot [name path=line 5] coordinates { (0,0) (0,0) };

\addplot [name path=line 0] coordinates { (0,0) (0,0) };

\addplot [name path=line 1,thick,  smooth, black]  {arctanh(\x)} node[above =24pt,pos=.63]{${\rm arctanh}(y)$};
\addplot [name path=line 2, thick, smooth, black!30] {1.1*1.1*x/2} node[above=0,pos=.95]{${\sigma_{\rho}^2\over 2}  y$};

\fill [
red!100,
name intersections={of=line 1 and line 2,total=\t}]
\foreach \s in {1,...,\t}{
	(intersection-\s) circle (2pt)
};
\end{axis}
\end{tikzpicture}
\qquad 		
\captionsetup{width=0.90\textwidth}
\caption{Point of intersection of the graphs of the functions  ${\sigma_\rho^2 y /2}$ and ${{\rm arctanh}(y)}$ when $\sigma_\rho\leqslant \sqrt{2}$.}
\label{fig:one_intersections}
\end{minipage}\hfil
\begin{minipage}{0.5\textwidth}
\centering
\pgfkeys{/pgf/declare function={arctanh(\x) = 0.5*(ln(1+\x)-ln(1-\x));}}
\begin{tikzpicture}
\begin{axis}[
Axis Style,
xtick={
	-0.875, 0.0, 0.875
},
xlabel= {$y$},
xticklabels={
	$y_-$, $0$,  $y_+$, 
},
]
\addplot [name path=line 5] coordinates { (0,0) (0,0) };

\addplot [name path=line 0] coordinates { (0,0) (0,0) };

\addplot [name path=line 1, thick, smooth, black]  {arctanh(\x)} node[above =24pt,pos=.63]{${\rm arctanh}(y)$};
\addplot [name path=line 2, thick, smooth, black!30] {1.75*1.75*x/2} node[below =-0.2pt,pos=.95]{${\sigma_{\rho}^2\over 2}  y$};

\fill [
red!100,
name intersections={of=line 1 and line 2,total=\t}]
\foreach \s in {1,...,\t}{
	(intersection-\s) circle (2pt)
};
\end{axis}
\end{tikzpicture}
\captionsetup{width=0.90\textwidth}
\caption{Points of intersection of the graphs of the functions  ${\sigma_\rho^2 y /2}$ and ${{\rm arctanh}(y)}$ when $\sigma_\rho> \sqrt{2}$.}
\label{fig:two_intersections}
\end{minipage}	
\end{figure}
\noindent
Observe that there always exists at least one solution for the above equation.  The number of solutions
of \eqref{tanh-equation} depends on whether ${\sigma_\rho}$ is larger or smaller than $\sqrt{2}$. If $\sigma_\rho\leqslant\sqrt{2}$, then \eqref{tanh-equation} has a unique solution, given by $y_0 = 0$ (Figure \ref{fig:one_intersections}), and if $\sigma_\rho>\sqrt{2}$, then \eqref{tanh-equation} has two additional non-trivial solutions (Figure \ref{fig:two_intersections}), $y_-=y_{-}(\sigma_\rho)<0$ and $y_+=y_{+}(\sigma_\rho)>0$ (which depend on $\sigma_\rho$), with $y_{+}=-y_{-}$; this means that 
$f_W(w;\sigma_\rho)$ has one critical point at $w_0=1/2$,  or three critical points, $w_-=A^{-1}(\sigma_\rho y_-)$, $w_0=1/2$ and $w_+=A^{-1}(\sigma_\rho y_+)$. 
As $\lim_{w\to 0^+} f_W(w;\sigma_\rho)=\lim_{w\to 1^-} f_W(w;\sigma_\rho)=0$, 
we then obtain the following result.
\begin{theorem}
	Let $W\sim {\rm ULS}(\sigma_\rho,g_c)$, with $g_c(x)=\exp(-x/2)$.  Then, the following holds: 
	 
	If $\sigma_\rho\leqslant\sqrt{2}$, then the ULS PDF $f_W(w;\sigma_\rho)$ is unimodal, with mode $w_0=1/2$.
	On the other hand,  
	if $\sigma_\rho>\sqrt{2}$, then the ULS PDF $f_W(w;\sigma_\rho)$ is bimodal, with modes 
	\begin{align*}
	w_-=A^{-1}(\sigma_\rho y_-)={1\over 1+\exp(-\sigma_\rho^2 y_{-})}
	\quad \text{and} \quad 
	w_+=A^{-1}(\sigma_\rho y_+)={1\over 1+\exp(-\sigma_\rho^2 y_{+})},
	\end{align*} 
	and minimum point  $w_0=1/2$, so that $0<w_-<w_0<w_+<1$.
Moreover, the graph of $f_W(w;\sigma_\rho)$ is symmetric around $w_0=1/2$  (Figure \ref{figpdf:1} a).
\end{theorem}

\subsubsection{Unit-log-Student-$t$}
In this case, from Subsubsection \ref{Unit-log-Student-t distribution}, we have $Z_\rho\sim t_\nu$, and $r(z)=(\nu+1)z/(\nu+z^2)$. So, from  \eqref{eq-critical point}, the equation for critical points of the ULS PDF is given by
\begin{align*}
\tanh\biggl({\sigma_\rho^2 \over 2}\, y\biggr)
= 
{(\nu+1) y\over \nu+\sigma_\rho^2 y^2},
\quad \text{with} \ y={A(w)/ \sigma_\rho},
\end{align*}
or, equivalently,
\begin{align}\label{tanh-equation-1}
\tanh(z)
= 
{2(\nu+1) z\over \nu\sigma_\rho^2+4z^2},
\quad \text{with} \ z=\sigma_\rho^2y/2.
\end{align}
Note that $w_0=1/2$ is a critical point of the ULS PDF because $z_0=0$ satisfies \eqref{tanh-equation-1}.
\pgfkeys{
	/pgfplots/Axis Style/.style={
		width=9.9cm, height=6cm, 
		axis x line=middle,
		axis y line=middle,
		yticklabel style={inner ysep=0pt, anchor=east},
		ytick={-1,0,1},
		samples=100,
		ymin=-2.0, ymax=2.0,
		xmin=-6.0, xmax=6.0,
		domain=-6:6
	}
}
\begin{figure}[!h]
	\centering
	\begin{minipage}{0.5\textwidth}
		\centering
		\begin{tikzpicture}
		\begin{axis}[
		Axis Style,
		xtick={
			 0.0
		},
		xlabel= {$z$},
		]
		\addplot [name path=line 5] coordinates { (0,0) (0,0) };
		
		\addplot [name path=line 0] coordinates { (0,0) (0,0) };
		
		\addplot [name path=line 1,thick,  smooth, black]  {tanh(\x)} node[above =-2.0pt,pos=.89]{${\rm tanh}(z)$};
		\addplot [name path=line 2, thick, smooth, black!30] {2*(5+1)*x/(5*1.8*1.8+4*x*x)} node[right =15.pt,below=-24.0,pos=.25]{${2(\nu+1) z\over \nu\sigma_\rho^2+4z^2}$};
		
		\fill [
		red!100,
		name intersections={of=line 1 and line 2,total=\t}]
		\foreach \s in {1,...,\t}{
			(intersection-\s) circle (2pt)
		};
		\end{axis}
		\end{tikzpicture}
		\qquad 		
		\captionsetup{width=0.90\textwidth}
		\caption{
		The graphs of the functions ${2(\nu+1) z/(\nu\sigma_\rho^2+4z^2)}$ and ${{\rm tanh}(z)}$ have a single point of intersection.}
		\label{fig:one_intersections-1}
	\end{minipage}\hfil
	\begin{minipage}{0.5\textwidth}
		\centering
		\begin{tikzpicture}
		\begin{axis}[
		Axis Style,
		xtick={
			-2.77, 0.0,  2.77
		},
		xlabel= {$z$},
		xticklabels={
			$z_-$,  $0$,  $z_+$,
		},
		]
		\addplot [name path=line 5] coordinates { (0,0) (0,0) };
		
		\addplot [name path=line 0] coordinates { (0,0) (0,0) };
		
		\addplot [name path=line 1, thick, smooth, black]  {tanh(\x)} node[above =-2.0pt,pos=.89]{${\rm tanh}(z)$};
		\addplot [name path=line 2, thick, smooth, black!30] {2*(5+1)*x/(5*0.75*0.75+4*x*x)} node[below =2pt,pos=.80]{${2(\nu+1) z\over \nu\sigma_\rho^2+4z^2}$};
		
		\fill [
		red!100,
		name intersections={of=line 1 and line 2,total=\t}]
		\foreach \s in {1,...,\t}{
			(intersection-\s) circle (2pt)
		};
		\end{axis}
		\end{tikzpicture}
		\captionsetup{width=0.90\textwidth}
		\caption{The graphs of the functions ${2(\nu+1) z/(\nu\sigma_\rho^2+4z^2)}$ and ${{\rm tanh}(z)}$ have three points of intersection.}
		\label{fig:three_intersections}
	\end{minipage}	
\end{figure}
\noindent

Because $\lim_{z\to\pm \infty} \tanh(z) = \pm 1$, by analyzing
the graphs of functions $z\mapsto\tanh(z)$ and $2(\nu+1) z/ (\nu\sigma_\rho^2+4z^2)$, we see that \eqref{tanh-equation-1} has one solution, given by $z_0=0$, or three  solutions, $z_-< z_0=0<z_+=-z_-$, that depend on the choice of parameters $\nu$ and $\sigma_\rho$ (Figures \ref{fig:one_intersections-1} and \ref{fig:three_intersections}). In other words, $f_W(w;\sigma_\rho)$ has one critical point at $w_0=1/2$, or three critical points, $
0<w_-=A^{-1}(2z_-/\sigma_\rho)<w_0=1/2<w_+=A^{-1}(2z_+/\sigma_\rho)<1$. 
As $\lim_{w\to 0^+} f_W(w;\sigma_\rho)=\lim_{w\to 1^-} f_W(w;\sigma_\rho)=\infty$ and $w_0=1/2$ is a critical point (and point of symmetry, see Proposition \ref{point of symmetry}) of $f_W(w;\sigma_\rho)$, we readily have the following result.
\begin{theorem}
	If $W\sim {\rm ULS}(\sigma_\rho,g_c)$, with $g_c(x)=(1+{x/\nu})^{-(\nu+2)/ 2}$, then the ULS PDF $f_W(w;\sigma_\rho)$ has a bathtub shape with minimum point
	$w_0=1/2$, or
	is decreasing-increasing-decreasing-increasing with 
	minimum points 
		\begin{align*}
	w_-=A^{-1}(2z_-/\sigma_\rho)={1\over 1+\exp(-2z_-)}
	\quad \text{and} \quad 
	w_+=A^{-1}(2z_+/\sigma_\rho)={1\over 1+\exp(-2z_+)},
	\end{align*} 
	and maximum point  $w_0=1/2$.
	Moreover, the graph of the  ULS PDF is symmetric with around $w_0=1/2$ (Figures \ref{figpdf:1} c, d).
\end{theorem}

\subsubsection{Unit-log-Laplace}\label{Unit-log-Laplace-1}
From Subsubsection \ref{Unit-log-Laplace}, $Z_\rho\sim {\rm Laplace}(0,1/\sqrt{2})$, and $r(z)=\sqrt{2} z/\vert z\vert$, $z\neq 0$.  The function $f_W(w,\sigma_{\rho})$ is not differentiable at $w_0=1/2$, and this point is a critical point with $f_W(w_0,\sigma_{\rho})=4/\sigma_\rho$. 
So, Equation \eqref{eq-critical point} for critical points (excluding $w_0$) of the ULS PDF is given by
\begin{align}\label{eq-critical point-Laplace}
\tanh\biggl({\sigma_\rho^2 \over 2}\, y\biggr)
=
{\sqrt{2}\over \sigma_\rho}\, {y\over \vert y\vert}
\quad \text{with} \ y={A(w)/ \sigma_\rho},
\end{align}
which implies that
$
y_{\pm}=\pm {2} \,{\rm arctanh}({\sqrt{2}/ \sigma_{\rho}})/\sigma_{\rho}^2.
$
In other words, $f_W(w;\sigma_\rho)$ has one critical point at $w_0=1/2$, or three critical points, 
$0<w_-=A^{-1}(\sigma_\rho y_-)<w_0=1/2<w_+=A^{-1}(\sigma_\rho y_+)<1$. 
Because
\begin{align*}
	\lim_{w\to 0^+}f_W(w,\sigma_{\rho})=\lim_{w\to 1^-}f_W(w,\sigma_{\rho})=
\begin{cases}
0, & \sigma_{\rho}<\sqrt{2},
\\
1/\sqrt{2}, & \sigma_{\rho}=\sqrt{2},
\\
\infty, & \sigma_{\rho}>\sqrt{2},
\end{cases}	
\end{align*}
and $w_0=1/2$ is a critical point of symmetry (Proposition \ref{point of symmetry}) of $f_W(w;\sigma_\rho)$, we readily obtain the following result.
\begin{theorem}
	Let $W\sim {\rm ULS}(\sigma_\rho,g_c)$, with $g_c(x)=K_0(\sqrt{2x})$.  Then, the following holds: 
	
	If $\sigma_\rho\leqslant\sqrt{2}$ then the ULS PDF $f_W(w;\sigma_\rho)$ is unimodal, with mode $w_0=1/2$.
	On the other hand,  
	if $\sigma_\rho>\sqrt{2}$ then the ULS PDF $f_W(w;\sigma_\rho)$ is decreasing-increasing-decreasing-increasing with 
	minimum points 
	\begin{align*}
	&w_-=A^{-1}(\sigma_\rho y_-)={1\over 1+\exp({2}\, {\rm arctanh}({\sqrt{2}/ \sigma_{\rho}}))},
	\\[0,2cm]
	&w_+=A^{-1}(\sigma_\rho y_+)={1\over 1+\exp(-{2}\, {\rm arctanh}({\sqrt{2}/ \sigma_{\rho}}))},
	\end{align*} 
	and maximum point  $w_0=1/2$, so that $0<w_-<w_0<w_+<1$.
	Moreover,  $f_W(w;\sigma_\rho)$ is not differentiable and symmetric at $w_0=1/2$ (Figure \ref{figpdf:1} b).
\end{theorem}

\subsection{Quantile function} \label{Quantile function}
From \eqref{general cdf formula}, the $p$-quantile function $Q_W(p)$ of $W\sim {\rm ULS}(\sigma_\rho,g_c)$, for $p\in(0,1)$, satisfies
\begin{align*}
	p=F_W(Q_W(p);\sigma_\rho)
=
	\mathbb{P}(
	Z_\rho
	\leqslant 
	A(Q_W(p))
	),
\end{align*}
where $Z_\rho$ is as in \eqref{sigma-and-Z-rho}.
Hence,
\begin{align}\label{eq:qf}
	Q_{Z_\rho}(p)=A(Q_W(p))
	\quad \Longleftrightarrow \quad
	Q_W(p)=A^{-1}(Q_{Z_\rho}(p))	\stackrel{\eqref{inverse-A}}{=}{1\over 1+\exp(-\sigma_\rho Q_{Z_\rho}(p))}.
\end{align}

As the quantiles of Gaussian, Student-$t$, and Laplace variables are well-known, by using the results of Subsubsections \ref{Unit-log-normal distribution}, \ref{Unit-log-Student-t distribution}, and \ref{Unit-log-Laplace} on the distribution of $Z_\rho$, we have the quantiles as presented in Table \ref{table:quantile}.
   \begin{table}[H]
	\caption{Density generators $(g_c)$ and $p$-quantiles of some ULS distributions.}
	\vspace*{0.15cm}
	\centering
	\resizebox{\linewidth}{!}{
		\begin{tabular}{llll}
			\hline
			Distribution &  $g_c$ & $Q_W(p)$
			\\ [0.5ex]
			\noalign{\hrule height 1.0pt}
			Unit-log-normal
			&  $\exp(-x/2)$ & 
			$[1+\exp(-\sigma_\rho \Phi^{-1}(p))]^{-1}$
			\\ [1ex]
			Unit-log-Student-$t$
			& $(1+{x\over\nu})^{-(\nu+2)/ 2}$  &  
$[1+\exp(-\sigma_\rho F_\nu^{-1}(p))]^{-1}$
			\\ [1ex]
			Unit-log-Laplace
			& $K_0(\sqrt{2x})$ & 
$[1+(2p)^{-\sigma_\rho/\sqrt{2}}]^{-1} \mathds{1}_{(0,1/2]}(p)
+
[1+(2-2p)^{\sigma_\rho/\sqrt{2}}]^{-1} \mathds{1}_{[1/2, \infty)}(p)$
			\\ [1ex]
			\hline
		\end{tabular}
	}
	\label{table:quantile} 
\end{table}
\noindent
Here, $\Phi^{-1}_\nu$ and $F^{-1}_\nu$ denote the quantiles of the standard normal and Student-$t$ distributions, respectively.

\subsection{Moments}\label{Moments}
For $W\sim {\rm ULS}(\sigma_\rho,g_c)$, since $W=T_1/(T_1+T_2)$ with $(T_1,T_2)\sim {\rm BLS}(\boldsymbol{\theta},g_c)$, it is clear that $0\leqslant\mathbb{E}(W^r)\leqslant 1$, $r>0$. 
Consequently, all positive moments of $W$ always exist.

In general, the moments of $W$ admit the following representation:
\begin{align*}
	\mathbb{E}(W^r)
	=
	\mathbb{E}(
	[A^{-1}( Z_\rho)]^r
	)
	\stackrel{\eqref{inverse-A}}{=}
	\mathbb{E}\left[
	{1\over (1+\exp(-\sigma_\rho Z_\rho))^r}
	\, 
	\right], \quad r\in\mathbb{R},
\end{align*}
where $Z_\rho$ is as given in \eqref{sigma-and-Z-rho}.

Upon using binomial expansion, the negative integer moments of $W$ can be written as
\begin{align*}
\mathbb{E}(W^{-n})
=
\mathbb{E}\left[
\left(
{1+\exp(-\sigma_\rho Z_\rho)}
\right)^n
\, 
\right]
=
\sum_{k=0}^{n}\binom{n}{k} M_{Z_\rho}(-k\sigma_\rho), \quad n\in\mathbb{N},
\end{align*}
where $M_{Z_\rho}(t)$ is the moment generating function (MGF) of $Z_\rho$.
%
%

As the MGFs of Gaussian and Laplace variables are avaliable in some simple explicit forms, by using the results of Subsubsections \ref{Unit-log-normal distribution} and \ref{Unit-log-Laplace} on the distribution of $Z_\rho$, we have the  negative moments corresponding to unit-log-normal and unit-log-Laplace (Table \ref{table:moments}).
On the other hand, since $Z_\rho\sim t_\nu$ as $g(x)=(1+{x/\nu})^{-(\nu+2)/ 2}$ (Subsubsection \ref{Unit-log-Student-t distribution}), it is clear that the  negative moments
corresponding to unit-log-Student-$t$ do not exist.
   \begin{table}[H]
	\caption{Density generators $(g_c)$ and negative moments for some ULS distributions.}
	\vspace*{0.15cm}
	\centering
		\begin{tabular}{lllllll}
			\hline
			Distribution &  $g_c$ & $\mathbb{E}(W^{-n})$ & Restriction
			\\ [0.5ex]
			\noalign{\hrule height 1.0pt}
			Unit-log-normal
			&  $\exp(-x/2)$ & 
			$\sum_{k=0}^{n}\binom{n}{k} \exp({1\over 2}\, k^2\sigma_\rho^2)$ & -
			\\ [1ex]
			Unit-log-Student-$t$
			& $(1+{x\over\nu})^{-(\nu+2)/ 2}$  &  
			$\nexists$ & -
			\\ [1ex]
			Unit-log-Laplace
			& $K_0(\sqrt{2x})$ & 
			$\sum_{k=0}^{n}\binom{n}{k} (1-{1\over 2} \, k^2\sigma_\rho^2)^{-1}$  & $n<\sqrt{2}/\sigma_\rho$, $\sigma_\rho\leqslant\sqrt{2}/2$
			\\ [1ex]
			\hline
		\end{tabular}
	\label{table:moments} 
\end{table}

\subsection{Shannon entropy}\label{Shannon entropy}

The entropy of a random variable can be interpreted as the average
level of uncertainty inherent in the possible outcomes of the variable.
Following \cite{Shannon1949}, the entropy $H$ of a random variable $W\sim {\rm ULS}(\sigma_\rho,g_c)$, which takes
values in the interval $(0,1)$ and is distributed according to
$f_W(\cdot,\sigma_{\rho})$ in \eqref{general pdf formula}, is given by
\begin{align*}
	H(W)=-\mathbb{E}[\log f_W(W,\sigma_{\rho})].
\end{align*} 

Upon using \eqref{general pdf formula} and the representation in \eqref{rep-stoch-biv-student-t}, we observe that
\begin{align*}
H(W)
=
2\mu+\log(\sigma_\rho)
+
\sigma[(1+\rho)\mathbb{E}(Z_1)+\sqrt{1-\rho^2}\mathbb{E}(Z_2)]
+
H(Z_\rho),
\end{align*}
where $Z_1$ and $Z_2$ are as given in Subsection \ref{Characterization of the ULS distribution} and $Z_\rho$ is as in \eqref{sigma-and-Z-rho}. Because $Z_1$ and $Z_2$ are identically distributed \cite[Proposition 1; Item ii, of][]{Saulo2022}, the above identity can be expressed as
\begin{align*}
H(W)
=
2\mu+\log(\sigma_\rho)
+
\sigma(1+\rho+\sqrt{1-\rho^2})\mathbb{E}(Z_1)
+
H(Z_\rho).
\end{align*}

It is known that the entropies of Gaussian, Student-$t$, and Laplace variables all exist. Then, by using the results of Subsubsections \ref{Unit-log-normal distribution}, \ref{Unit-log-Student-t distribution}, and \ref{Unit-log-Laplace} on the distribution of $Z_\rho$, we obtain the entropies for the three models as presented in \ref{table:entropy}.
%
%
%
%
%
%
   \begin{table}[H]
	\caption{Density generators $(g_c)$ and entropies for some ULS distributions.}
	\vspace*{0.15cm}
	\centering
	    		\resizebox{\linewidth}{!}{
	\begin{tabular}{llll}
		\hline
		Distribution &  $g_c$ & $H(W)$
		\\ [0.5ex]
		\noalign{\hrule height 1.0pt}
		Unit-log-normal
		&  $\exp(-x/2)$ & 
		$2\mu+\log(\sigma_\rho)
		+
		{1\over 2}\, [\log(2\pi)+1]$
		\\ [1ex]
		Unit-log-Student-$t$
		& $(1+{x\over\nu})^{-(\nu+2)/ 2}$  &  
		$2\mu+\log(\sigma_\rho)
		+
		{\nu+1\over 2}\, \left[\Psi\left({\nu+1\over 2}\right)-\Psi\left({\nu\over 2}\right)\right]
		+
		\log\left[\sqrt{\nu}\, {\rm B}\left({\nu\over 2}, {1\over 2}\right)\right]$
		\\ [1ex]
		Unit-log-Laplace
		& $K_0(\sqrt{2x})$ & 
		$2\mu+\log(\sigma_\rho)
		+
		\log(\sqrt{2}{\rm e})$
		\\ [1ex]
		\hline
	\end{tabular}
	}
	\label{table:entropy} 
\end{table}
\noindent
Here,  $\Psi(x)=\Gamma'(x)/\Gamma(x)$ is the digamma function and ${\rm B}(x,y)=\Gamma(x)\Gamma(y)/\Gamma(x+y)$ is the complete beta function, with $\Gamma(t)=\int_0^\infty x^{t-1} \exp(-x) \,{\rm d}x$, $t>0$, being the gamma function.

\subsection{Maximum likelihood estimation}\label{ml_simulation_res}
\noindent

Let $\{W_{i}; i=1,\ldots,n\}$ be a univariate random sample of size $n$ from the ${\rm ULS}(\sigma_\rho,g_c)$ distribution with PDF as in \eqref{general pdf formula}, and let $w_{i}$ be the sample observations of $W_{i}$. Then, the log-likelihood function for $\sigma_\rho$ is given by
\begin{align*}
	\ell(\sigma_\rho)
	=
	-n\log(\sigma_\rho)
	-\sum_{i=1}^{n} \log[w_i(1-w_i)]
	+
	\sum_{i=1}^{n} 
	\log
	f_{Z_\rho}(A(w_i)),
\end{align*}
where the random variable $Z_\rho$ is as defined in \eqref{sigma-and-Z-rho} and $A$ is as in \eqref{def-A}.
In the case when a supremum $\widehat{\sigma}_\rho$ exists, it must be such that 
\begin{align}\label{likelihood equation}
	{\partial \ell(\sigma_\rho)\over\partial\sigma_\rho}
	\bigg\vert_{\sigma_\rho=\widehat{\sigma}_\rho}
	=0
\quad \text{and} \quad 
	{\partial^2 \ell(\sigma_\rho)\over\partial\sigma_\rho^2}
	\bigg\vert_{\sigma_\rho=\widehat{\sigma}_\rho}
	<0, 
\end{align}
with
\begin{align}\label{2-derivative}
	{\partial \ell(\sigma_\rho)\over\partial\sigma_\rho}
	&=
	-{n\over \sigma_\rho}
	-
	\frac{1}{\sigma_\rho^2}
	\sum_{i=1}^{n}
	\log\Big({w_i\over 1-w_i}\Big)
	G(A(w_i)), \nonumber
\\[0,2cm]
{\partial^2 \ell(\sigma_\rho)\over\partial\sigma_\rho^2}
&
=
-{1\over \sigma_\rho}\, \left[2\, {\partial\ell(\sigma_\rho) \over\partial\sigma_\rho}+{n\over\sigma_\rho }\right]
+
\frac{1}{\sigma_\rho^4}
\sum_{i=1}^{n}
\log^2\Big({w_i\over 1-w_i}\Big)
G'(A(w_i)).
\end{align}
We shall now adopt the notation
\begin{align}\label{def-x-rho}
	G(x)
	=
	f_{Z_\rho}'(x)/f_{Z_\rho}(x).
\end{align}
Observe now that the likelihood equation $	{\partial \ell(\sigma_\rho)/\partial\sigma_\rho}
\big\vert_{\sigma_\rho=\widehat{\sigma}_\rho}=0$ in \eqref{likelihood equation} can be written as 
\begin{align}\label{eq-sigma-rho}
	-\frac{1}{n}
	\sum_{i=1}^{n}
	\log\Big({w_i\over 1-w_i}\Big)
	G(A(w_i))
	\bigg\vert_{\sigma_\rho=\widehat{\sigma}_\rho}
	=
	\widehat{\sigma}_\rho.
\end{align}
Any nontrivial root $\widehat{\sigma}_\rho$ of the above equation is the ML estimate of ${\sigma}_\rho$ in the loose sense. When the parameter value provides the absolute maximum of the log-likelihood function, it is called an ML estimate in the strict sense.

\subsubsection{Unit-log-normal}
From Subsubsection \ref{Unit-log-normal distribution},  $Z_\rho\sim N(0,1)$. Then,  $G(x)$ in \eqref{def-x-rho} and its derivative are given by $G(x)=-x$ and $G'(x)=-1$. 
Thus, by using \eqref{eq-sigma-rho}, the ML estimate of ${\sigma}_\rho^2$ in the loose sense   is given by
\begin{align*}
	\widehat{\sigma}_\rho^2
	=
	{\frac{1}{n}
		\sum_{i=1}^{n}
		\log^2\Big({w_i\over 1-w_i}\Big)}.
\end{align*}
From  \eqref{2-derivative}, it is then easy to verify that ${\partial^2 \ell(\sigma_\rho)/\partial\sigma_\rho^2}
\big\vert_{\sigma_\rho=\widehat{\sigma}_\rho}<0$, and so $\widehat{\sigma}_\rho^2$ is an ML estimate in the strict sense.

Let $\widehat{\mathcal{O}}_\rho^2$ be the corresponding ML estimator of ${\sigma}_\rho^2$. By using Proposition \ref{Another stochastic representation},  we have $A(W)=Z_\rho\sim N(0,1)$, for $W\sim {\rm ULS}(\sigma_\rho,g_c)$. Then (for $k=1,2,\ldots$),
\begin{align}\label{Unbiasedness}
\mathbb{E}\big[(\widehat{\mathcal{O}}_\rho^2)^k\big]
=
\mathbb{E}\left[\log^{2k}\Big({W\over 1-W}\Big)\right]
=
{\sigma}_\rho^{2k}
\mathbb{E}\left[A^{2k}(W)\right] 
=
{\sigma}_\rho^{2k}
\mathbb{E}\left(Z_\rho^{2k}\right)
=
{\sigma}_\rho^{2k}\, {(2k)!\over 2^k k!}.
\end{align}
\begin{proposition}
Let $W_1, \ldots, W_n$ be independent and identically distributed random variables with the unit-log-normal distribution in \eqref{ULS-PDF-normal}, where $\sigma_\rho$ is unknown. Then, the ML estimator $\widehat{\mathcal{O}}_\rho^2$ of $\sigma_\rho^2$  is
\begin{enumerate}
\item unbiased;
\item consistent;
\item asymptotically normal; that is, $\sqrt{n}(\widehat{\mathcal{O}}_\rho^2- {\sigma}_\rho^2)$ converges in distribution to a normal distribution with
mean zero and variance $2\sigma_\rho^2$;
\item efficient.
\end{enumerate}
\end{proposition}
\begin{proof}
By taking $n=1$ in \eqref{Unbiasedness}, the unbiasedness of $\widehat{\mathcal{O}}_\rho^2$ follows.
	
By the Strong Law of Large Numbers and
the Central Limit Theorem, it follows that $\widehat{\mathcal{O}}_\rho^2$ is consistent and asymptotically normal.

Finally, since the variance of $\widehat{\mathcal{O}}_\rho^2$ coincides with the Cramér-Rao lower bound $[n I(\sigma_\rho)]^{-1}$, where
$
I(\sigma_\rho)=-\mathbb{E}[{\partial^2 \log f_W(w;\sigma_\rho) /\partial(\sigma_\rho^2)^2}]
=
{1/(2 \sigma_\rho^2)}
$
is the Fisher information in one observation from $f_W(w;\sigma_\rho)$, 
the efficiency of $\widehat{\mathcal{O}}_\rho^2$ follows readily.
\end{proof}

\subsubsection{Unit-log-Student-$t$}
From Subsubsection \ref{Unit-log-Student-t distribution},  $Z_\rho\sim t_\nu$. Then,  $G(x)$ in \eqref{def-x-rho} and its derivative are given by $G(x)=-(\nu+1)x/(\nu+x^2)$ and $G'(x)=-(\nu+1)(\nu-x^2)/(\nu+x^2)^2$. 
Thus, by using \eqref{eq-sigma-rho}, the ML estimate of ${\sigma}_\rho^2$ in the loose sense satisfies
\begin{align}\label{id-ML-est}
	\frac{\nu+1}{n}
	\sum_{i=1}^{n}
	\dfrac{\displaystyle \log^2\Big({w_i\over 1-w_i}\Big)}{\nu \widehat{\sigma}_\rho^2+\log^2\Big({\displaystyle w_i\over \displaystyle 1-w_i}\Big)}
	=1.
\end{align}

\begin{proposition}
Any ML estimate of ${\sigma}_\rho^2$ in the loose sense
is an ML estimate in the strict sense. 
\end{proposition}
\begin{proof}
Suppose $\widehat{\sigma}_\rho^2$ is an ML estimate of ${\sigma}_\rho^2$ in the loose sense. 
From \eqref{2-derivative}, it follows that $\widehat{\sigma}_\rho^2$ is an ML estimate in the strict sense whenever 
\begin{align*}
&{\nu+1\over n}
\sum_{i=1}^{n}
\log^2\Big({ w_i\over 1-w_i}\Big)\,
\dfrac{
\nu\widehat{\sigma}_\rho^2-\log^2\Big({\displaystyle w_i\over\displaystyle 1-w_i}\Big)
}{ 
\left[\nu\widehat{\sigma}_\rho^2+\log^2\Big({\displaystyle w_i\over\displaystyle 1-w_i}\Big)\right]^2
}
>-1.
\end{align*}
By using the partial fraction expansion
\begin{align*}
	{x-y\over (x+y)^2}={2x\over (x+y)^2}-{1\over x+y},
\end{align*}
the above inequality is written as
\begin{align*}
{	2\nu(\nu+1)\widehat{\sigma}_\rho^2\over n}
\sum_{i=1}^{n}
\dfrac{
\log^2\Big({\displaystyle w_i\over\displaystyle 1-w_i}\Big)
}{ 
	\left[\nu\widehat{\sigma}_\rho^2+\log^2\Big({\displaystyle w_i\over\displaystyle 1-w_i}\Big)\right]^2
}
-
{\nu+1\over n}
\sum_{i=1}^{n}
\dfrac{
\log^2\Big({\displaystyle w_i\over\displaystyle 1-w_i}\Big)
}{ 
	\nu\widehat{\sigma}_\rho^2+\log^2\Big({\displaystyle w_i\over\displaystyle 1-w_i}\Big)
}
>-1.
\end{align*}
As $\widehat{\sigma}_\rho^2$ is an ML estimate in the loose sense, it satisfies  \eqref{id-ML-est}, and consequently, the above inequality is true. This completes the proof of the result.
\end{proof}

\subsubsection{Unit-log-Laplace}

In this case, note that $f_W(w;\sigma_\rho)=\exp[\theta T(w)-\psi(\theta)] h(w)$, $W\sim {\rm ULS}(\sigma_\rho,g_c)$, with  $\theta=1/\sigma_\rho$,
\begin{align*}
T(w)=-{\sqrt{2}} \biggl\vert\log\Big({w\over 1-w}\Big)\biggr\vert,
\quad 
\psi(\theta)=-\log(\theta)
\quad \text{and} \quad 
h(w)={1\over \sqrt{2}w(1-w)}. 
\end{align*} 
Then, $f_W$ belongs to the one-parameter exponential family with $\Theta=\{\theta:\theta>0\}$ parameter space.

\begin{proposition}\label{est-suff}
\begin{enumerate}
\item For $\theta\in\Theta^{\rm o}$,  the interior of $\Theta$, all moments of $T(W)$ exist, and $\psi(\theta)$ is infinitely differentiable at any such $\theta$. Furthermore,
\begin{align*}
\mathbb{E}\left(\sqrt{2}\biggl\vert\log\Big({W\over 1-W}\Big)\biggr\vert\right)=\sigma_{\rho}
\quad \text{and} \quad	
{\rm Var}\left(\sqrt{2}\biggl\vert\log\Big({W\over 1-W}\Big)\biggr\vert\right)
=
\sigma_{\rho}^2;
\end{align*}

\item Given an i.i.d. sample of size $n$ from $f_W(w;\sigma_\rho)$,
\begin{align*}
	-{\sqrt{2}} \sum_{i=1}^{n} \biggl\vert\log\Big({W\over 1-W}\Big)\biggr\vert
\end{align*} 
is minimally sufficient;

\item The Fisher information function exists, is finite at all $\theta\in\Theta^{\rm o}$ and equals
$I (\theta)=[I(\sigma_\rho)]^{-1} = \sigma_{\rho}^{2}$.
\end{enumerate}
\end{proposition}
\begin{proof}
The proof follows by direct application of Proposition 16.1 of  \cite{dasgupta2008asymptotic}, and is therefore omitted for the sake of conciseness.
\end{proof}

\begin{proposition}\label{pro-uniq}
Let $\theta = \theta_0 \in \Theta^{\rm o}$ be the true value. Then, for all large $n$, with probability $1$, a unique ML estimator of
$\sigma_{\rho}$ exists, and is consistent and asymptotically normal.
\end{proposition}
\begin{proof}
As $\psi''(\theta)=\sigma_{\rho}^2>0$,
by Theorem 16.1 of \cite{dasgupta2008asymptotic}, a unique ML estimator of
$\theta=1/\sigma_{\rho}$ exists, and is consistent and asymptotically normal. Hence, by the invariance property of the ML estimator, by the continuous application theorem for convergence in probability and by delta method, the proof of the result gets completed.
\end{proof}

Furthermore, in the unit-log-Laplace case, the likelihood function for $\sigma_\rho$ is given by
\begin{align*}
L(\sigma_\rho)\stackrel{\eqref{ULS-PDF-Laplace}}{=}
\prod_{i=1}^{n}
{1\over w_i(1-w_i)\sigma_\rho}\, 
f_\ell(A(w_i)),
\end{align*}
where $A$ is as given in \eqref{def-A} and $f_\ell$ is the PDF of Laplace distribution with location parameter $0$ and scale parameter $1/\sqrt{2}$. Because $f_\ell(x)$ is not differentiable at $x=0$ we cannot use the log-likelihood function differentiation method as in the last two cases. In what follows, we maximize $L(\sigma_\rho)$.

Indeed, an algebraic manipulation yields
\begin{align*}
L(\sigma_\rho)
=
{1\over 2^{n/2}} 
\left[\prod_{i=1}^{n} {1\over w_i(1-w_i)}\right]
\dfrac{\displaystyle\exp\left(-{\sqrt{2}\over\sigma_\rho}\, 	\sum_{i=1}^{n}
	\left\vert\log\Big({w_i\over 1-w_i}\Big)\right\vert \right)}{\sigma_\rho^n}.
\end{align*}
As the function $\sigma_\rho\longmapsto\exp(-t/\sigma_\rho)/\sigma_\rho^n$, $t>0$, reaches a maximum at $\sigma_\rho=t/n$, the right-hand expression in the above equation is at most
\begin{align*}
{\exp(-n)\over 2^{n}} 
\left[\prod_{i=1}^{n} {1\over w_i(1-w_i)}\right] 
\Bigg/
\left[
{1\over n} \sum_{i=1}^{n} \biggl\vert\log\Big({w_i\over 1-w_i}\Big)\biggr\vert
\right]^{n}.
\end{align*}
So, the ML estimate of ${\sigma}_\rho^2$ is  
\begin{align*}
\widehat{\sigma}_\rho^2
=
{\sqrt{2}\over n} \sum_{i=1}^{n} \biggl\vert\log\Big({w_i\over 1-w_i}\Big)\biggr\vert.
\end{align*}
This corroborates the uniqueness of the  ML estimator evidenced in Proposition \ref{pro-uniq}. Moreover, from Item 1 of Proposition \ref{est-suff}, the  ML estimator of ${\sigma}_\rho^2$ is unbiased and consistent.


\section{Simulation results} \label{Sec:5}
\noindent

In this section, we present the results of Monte Carlo simulation studies for unit-log-normal model (the results of the unit-log-Student-$t$ and unit-log-Laplace are somewhat similar and so are omitted here due to space limitations), considering different scenarios of parameters and sample sizes. The first part of the study presents the evaluation of the ML estimates, while the second part evaluates the empirical distribution of the generalized Cox-Snell (GCS) and randomized quantile (RQ) residuals, which are given, respectively, by
\begin{equation*}
    \hat{r}_i^{\text{GCS}} = -\log[1-\widehat{F}_W(w_i;\widehat{\sigma}_{\rho})]\quad \text{and} \quad
     \hat{r}_i^{\text{RQ}} = \Phi^{-1}(\widehat{F}_W(w_i;\widehat{\sigma}_{\rho})), \quad  i = 1,\dots,n,
\end{equation*}
where $\widehat{F}_W$ is the fitted unit-log-normal or unit-log-Student-$t$ CDF, and $\widehat{\sigma}_{\rho}$ is the ML estimate of $\sigma_{\rho}$.  When the models are correctly specified, the GCS is asymptotically standard exponentially distributed, while the RQ is asymptotically standard normally distributed. The distributions of these residuals are important as they will be used to assess the goodness-of-fit of the considered model.

Both studies consider simulated data generated from each of the unit-log-normal and unit-log-Student-$t$ models according to the following stochastic relation:
\begin{equation*}
W
=
A^{-1}(F_W^{-1}(U))
\stackrel{\eqref{inverse-A}}{=}
{1\over 1+\exp(-\sigma_\rho F_W^{-1}(U))},
\end{equation*}
where ${F}_W$ is the unit-log-normal or unit-log-Student-$t$ CDF and $U\sim U[0,1]$. The
Monte Carlo simulation experiments were performed using the \texttt{R} environment; see
http://www.r-project.org.

\subsection{Maximum likelihood estimation} \label{Sec:5.1}

 The simulation scenario used the following settings: sample size $n \in {50,100,150,250,600}$ and true parameter $\sigma_{\rho} \in \{0.15,0.25,0.50,1.0,1.5\}$, with $500$ Monte Carlo replications for each sample size. To study the performance of the ML estimators, we  computed the bias and root mean square error (RMSE), defined by
\begin{eqnarray*}
 \widehat{\textrm{Bias}}(\widehat{\theta}) = \frac{1}{M} \sum_{i = 1}^{M} \widehat{\theta}^{(i)} - \theta, \quad
\widehat{\mathrm{RMSE}}(\widehat{\theta}) = {\sqrt{\frac{1}{M} \sum_{i = 1}^{M} (\widehat{\theta}^{(i)} - \theta)^2}}, \\
\end{eqnarray*}
where $\theta$ and $\widehat{\theta}^{(i)}$ are the true parameter value and its respective $i$-th ML estimate, and $M$ is the number of Monte Carlo replications.

The ML estimation results are shown in Figure \ref{fig_dagum_MC_BRC} wherein the empirical bias and RMSE are both presented. It can be seen that the simulations produced the expected outcomes, namely, that the bias and RMSE both approach zero as $n$ increases.

\begin{figure}[!ht]
\vspace{-0.25cm}
\centering
{\includegraphics[height=5.0cm,width=5.0cm]{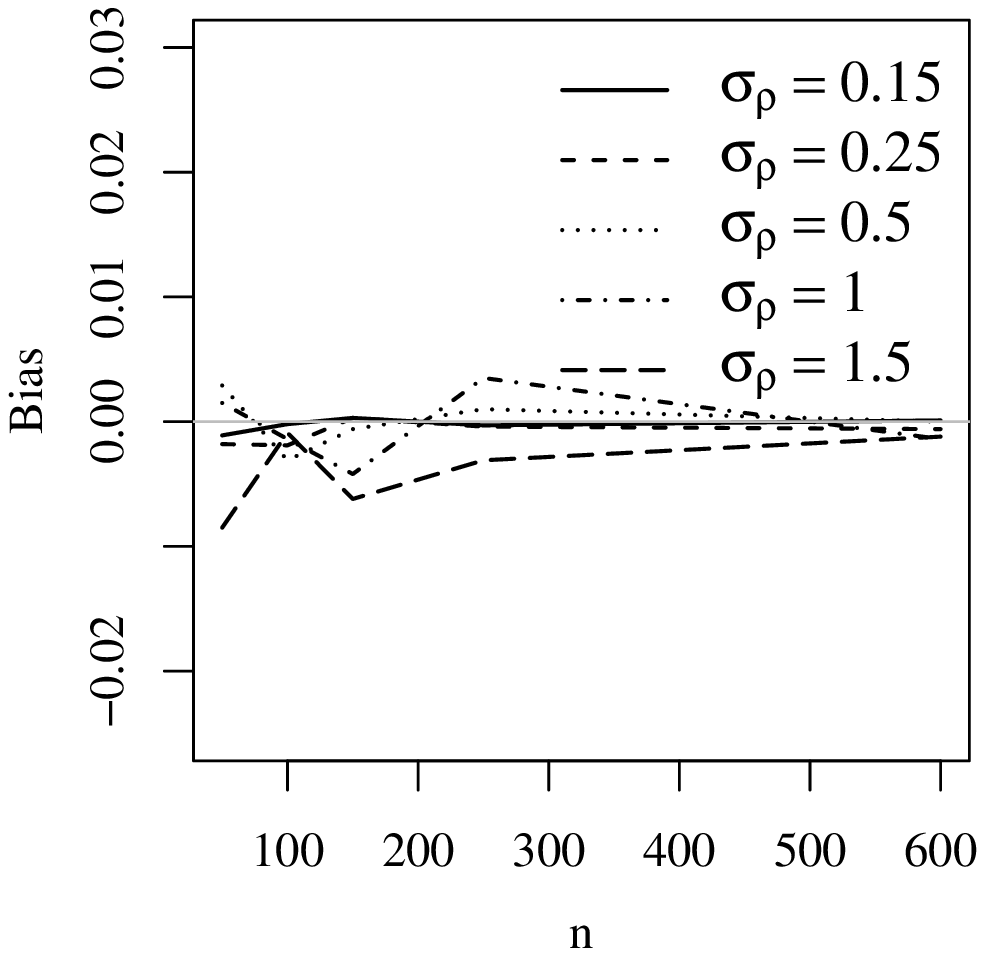}}\hspace{-0.25cm}
{\includegraphics[height=5.0cm,width=5.0cm]{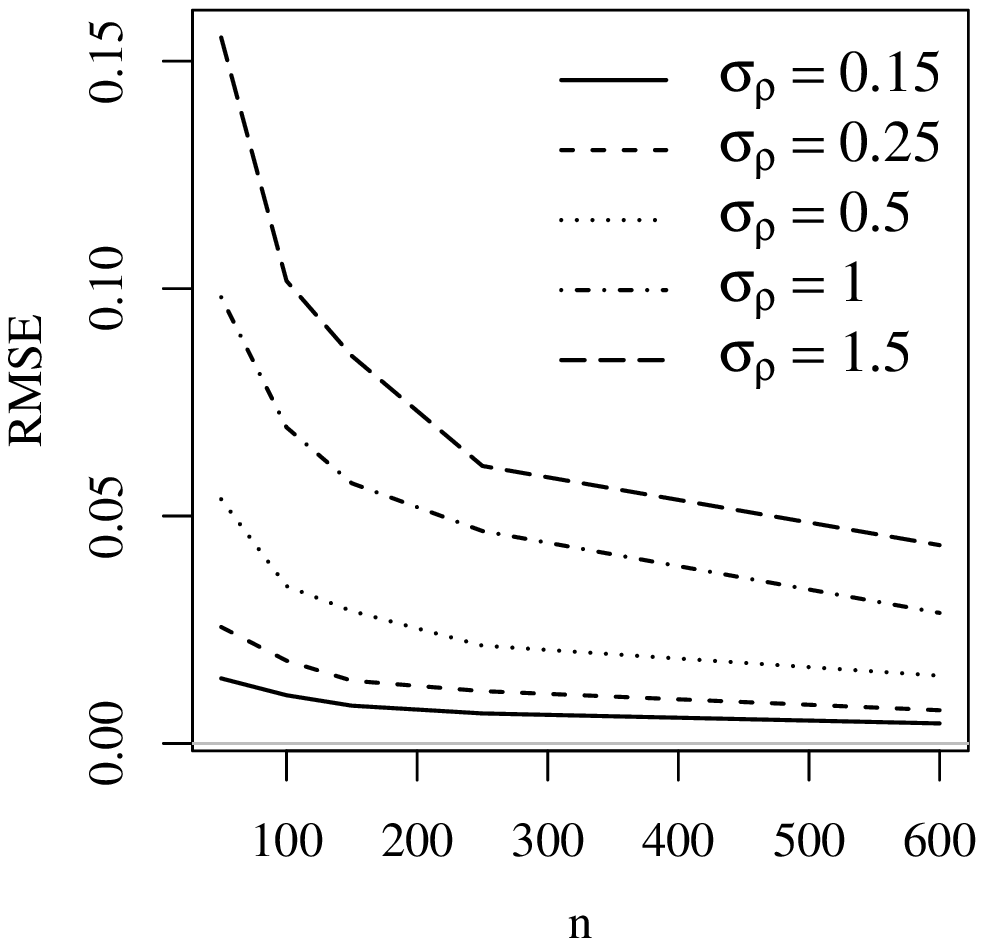}}\hspace{-0.25cm}
%
\vspace{-0.2cm}
\caption{Monte Carlo simulation results for the unit-log-normal model.}
\label{fig_dagum_MC_BRC}
\end{figure}

\subsection{Empirical distribution of residuals}\label{residuals_simulation_res}
\noindent

In this subsection, we evaluate the performance of GCS and RQ residuals. We present the empirical means of the following descriptive statistics: mean, median, standard deviation (Sd), coefficient of skewness and coefficient of kurtosis, whose values are expected to be 1, 0.69, 1, 2 and 6, respectively, for the GCS residuals, and 0, 0, 1, 0 and 0, respectively, for the RQ residuals. From Figures \ref{fig_normal_MC_GCS} and \ref{fig_normal_MC_RQ}, we note that the considered residuals conform well with their reference distributions, and we can therefore use the RQ and GCS residuals to verify well the fit of the proposed models.


%

\begin{figure}[!ht]
\vspace{-0.25cm}
\centering
{\includegraphics[height=4.8cm,width=4.8cm]{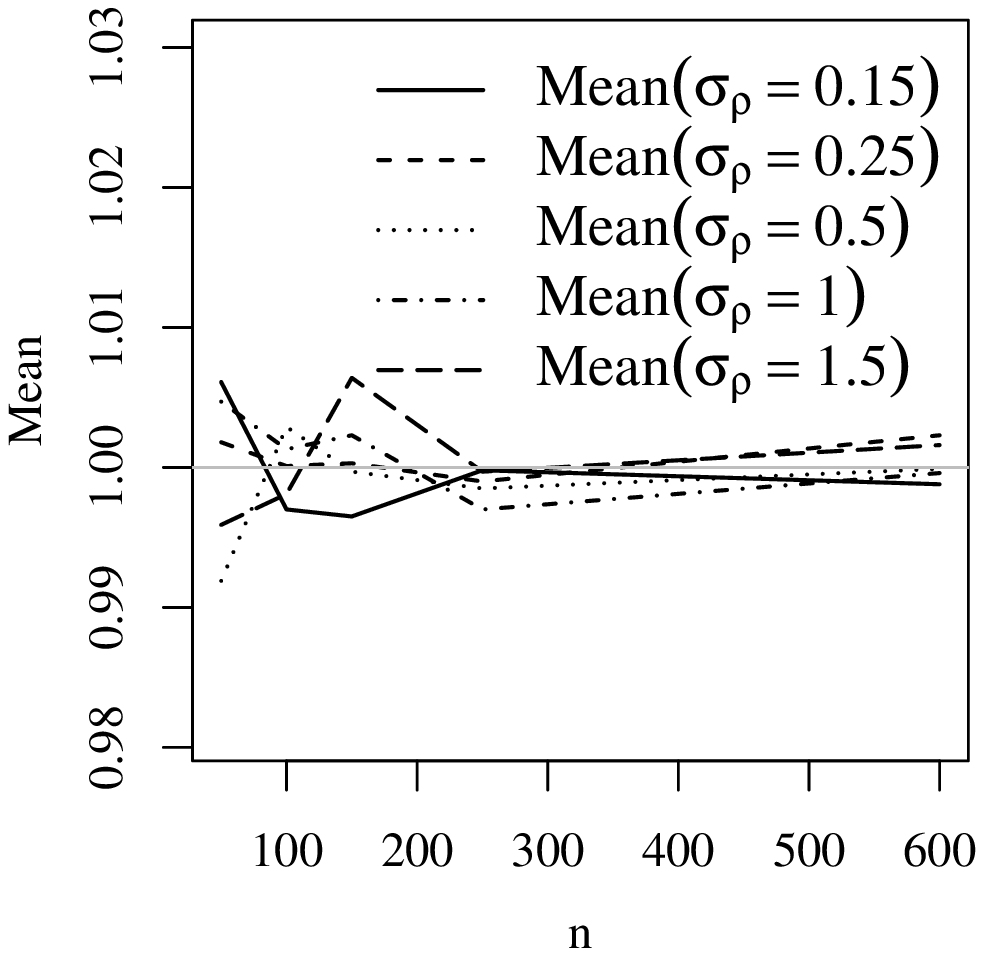}}\hspace{-0.25cm}
{\includegraphics[height=4.8cm,width=4.8cm]{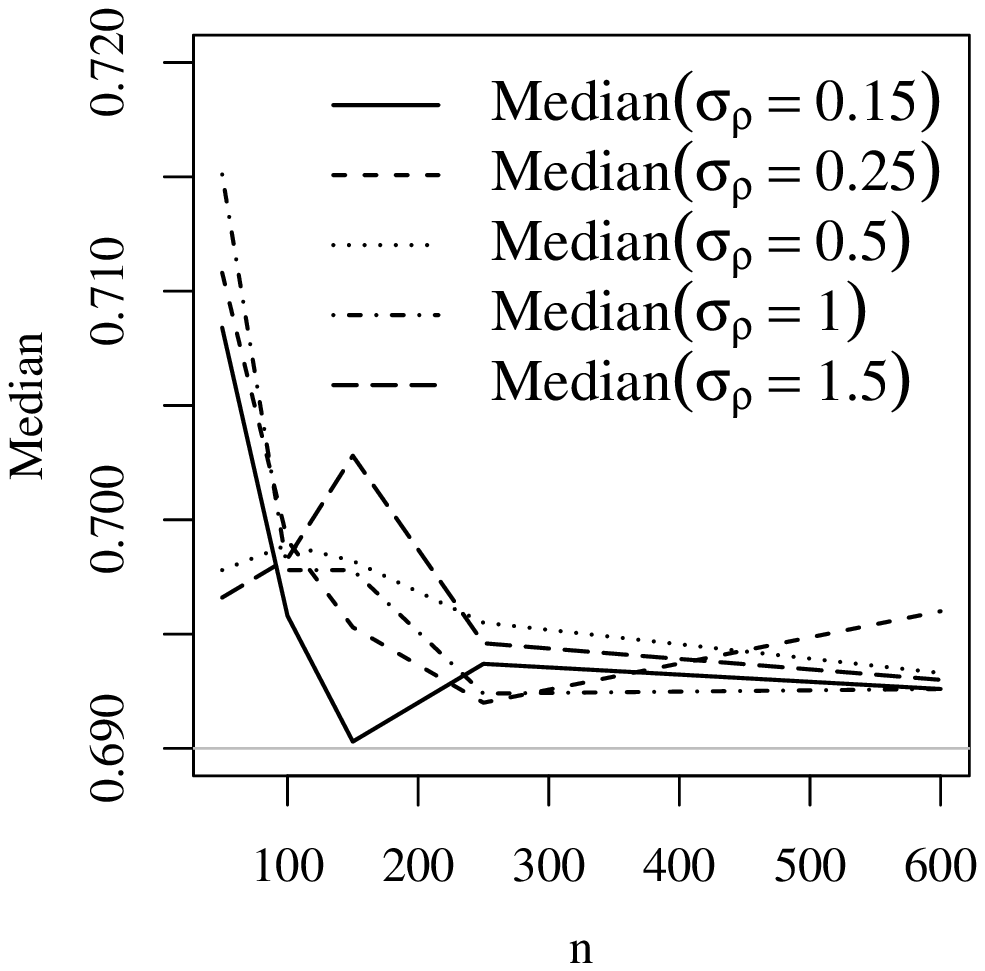}}\hspace{-0.25cm}
{\includegraphics[height=4.8cm,width=4.8cm]{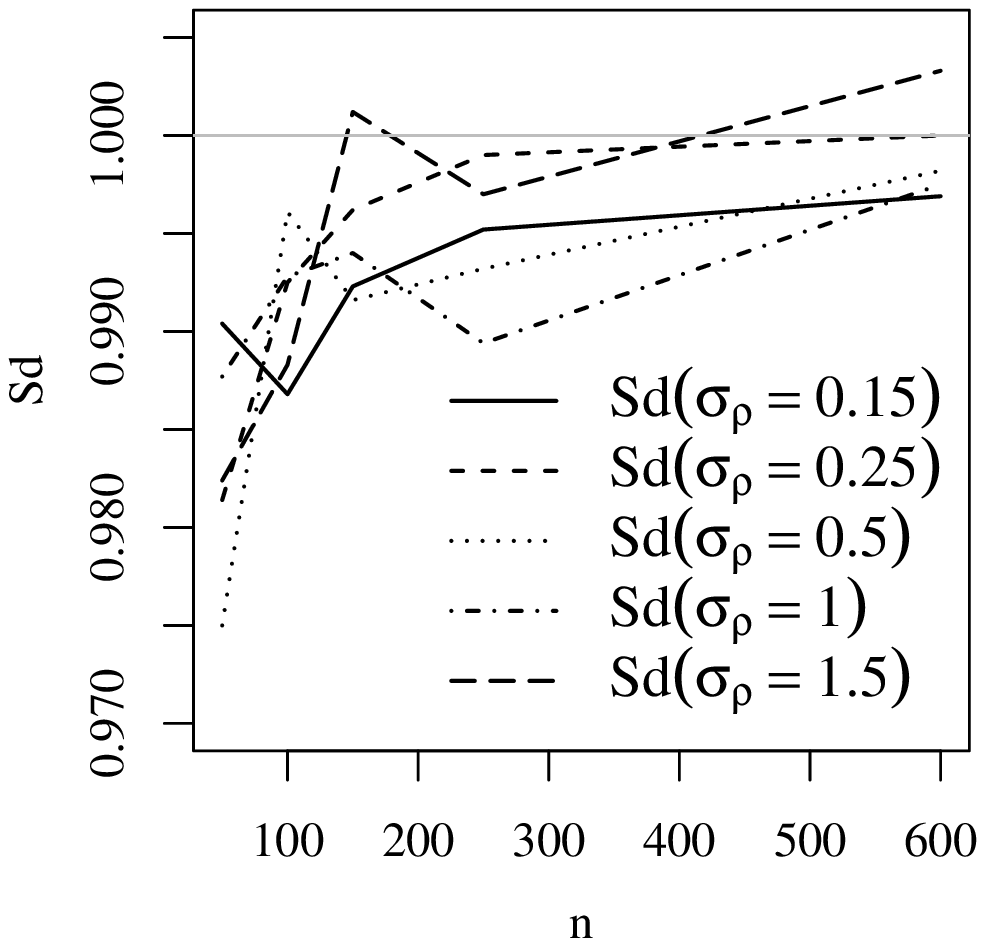}}\hspace{-0.25cm}
{\includegraphics[height=4.8cm,width=4.8cm]{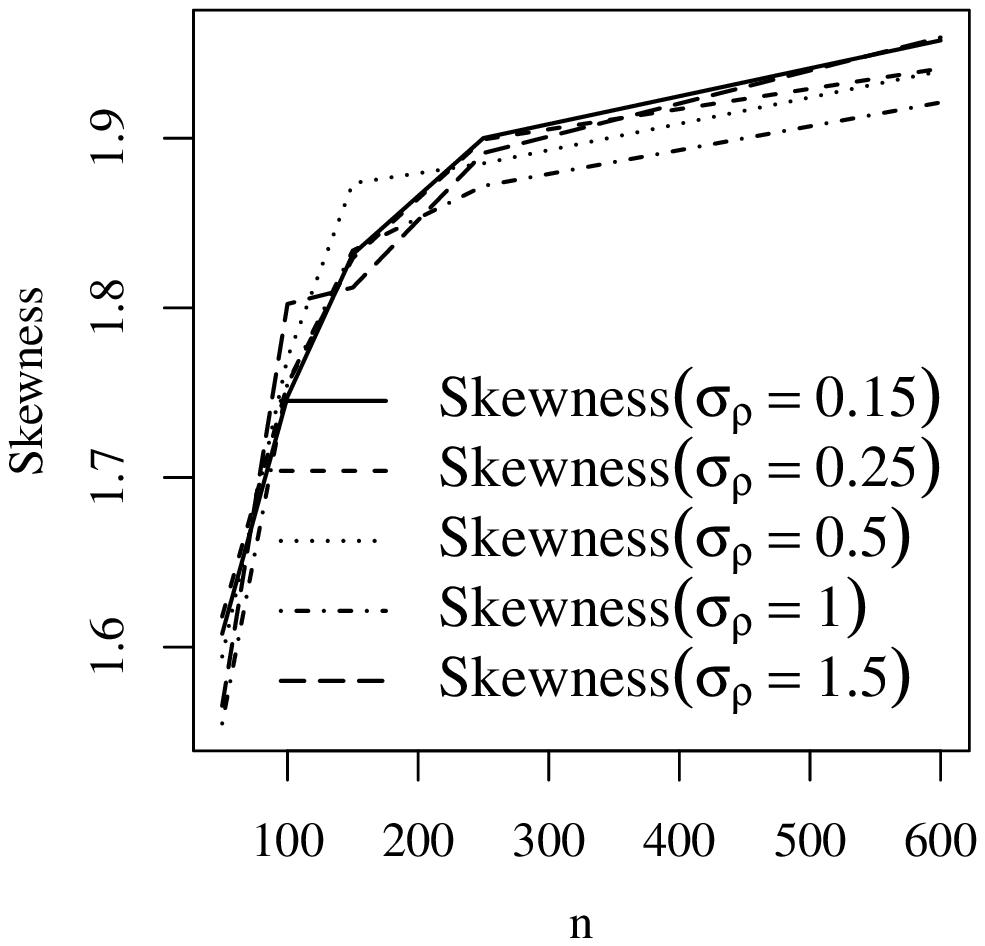}}\hspace{-0.25cm}
{\includegraphics[height=4.8cm,width=4.8cm]{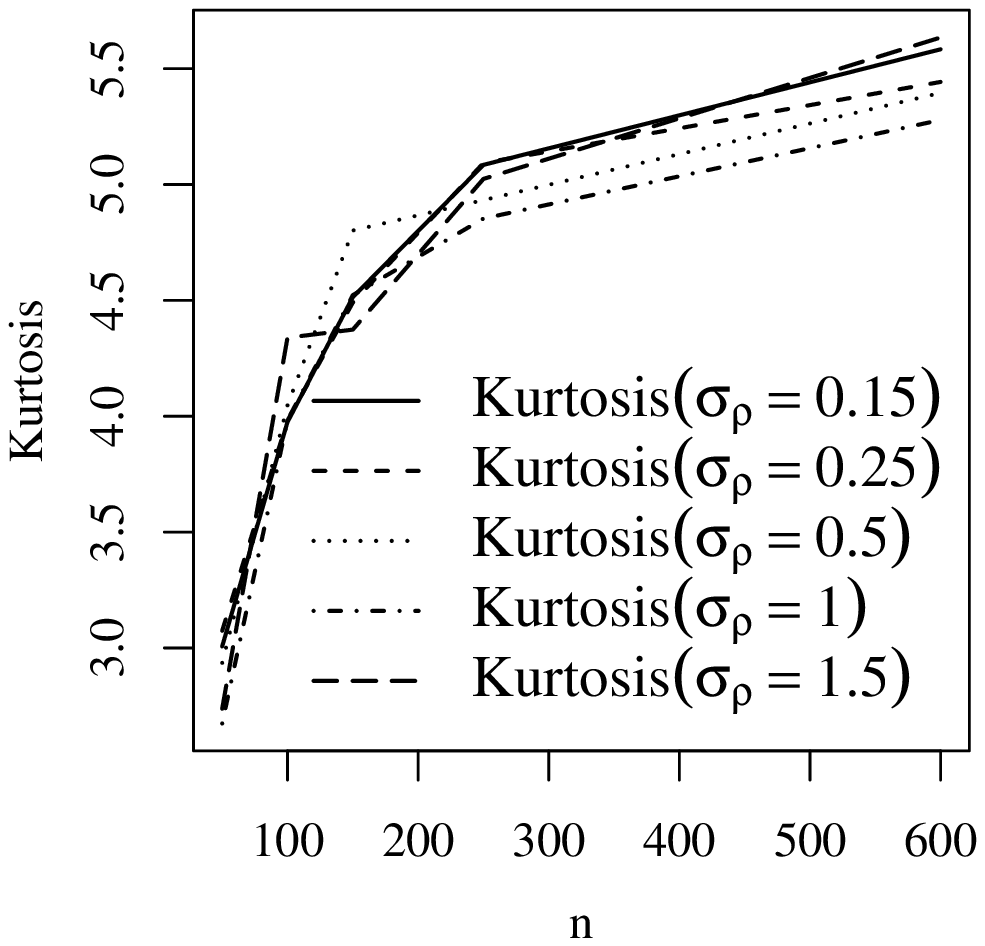}}

\vspace{-0.2cm}
\caption{Monte Carlo simulation results of the RQ residuals for the unit-log-normal model.}
\label{fig_normal_MC_GCS}
\end{figure}

\begin{figure}[!ht]
\vspace{-0.25cm}
\centering
{\includegraphics[height=4.8cm,width=4.8cm]{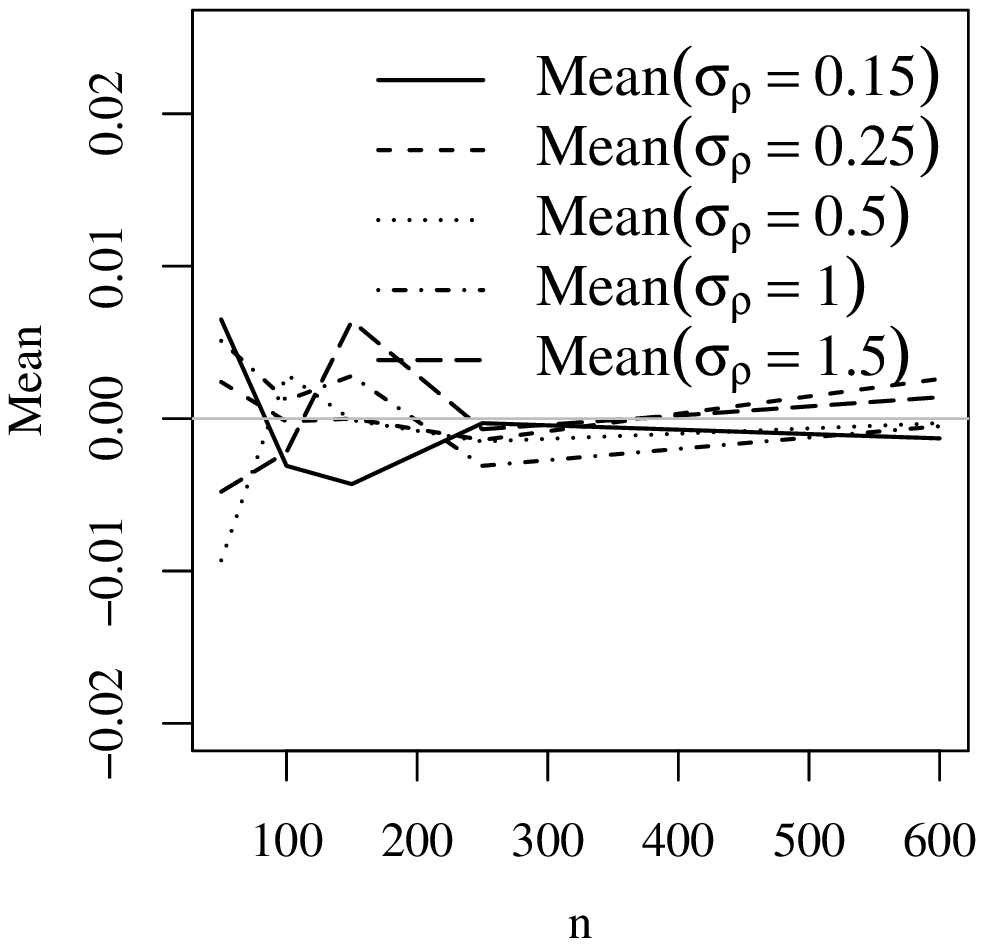}}\hspace{-0.25cm}
{\includegraphics[height=4.8cm,width=4.8cm]{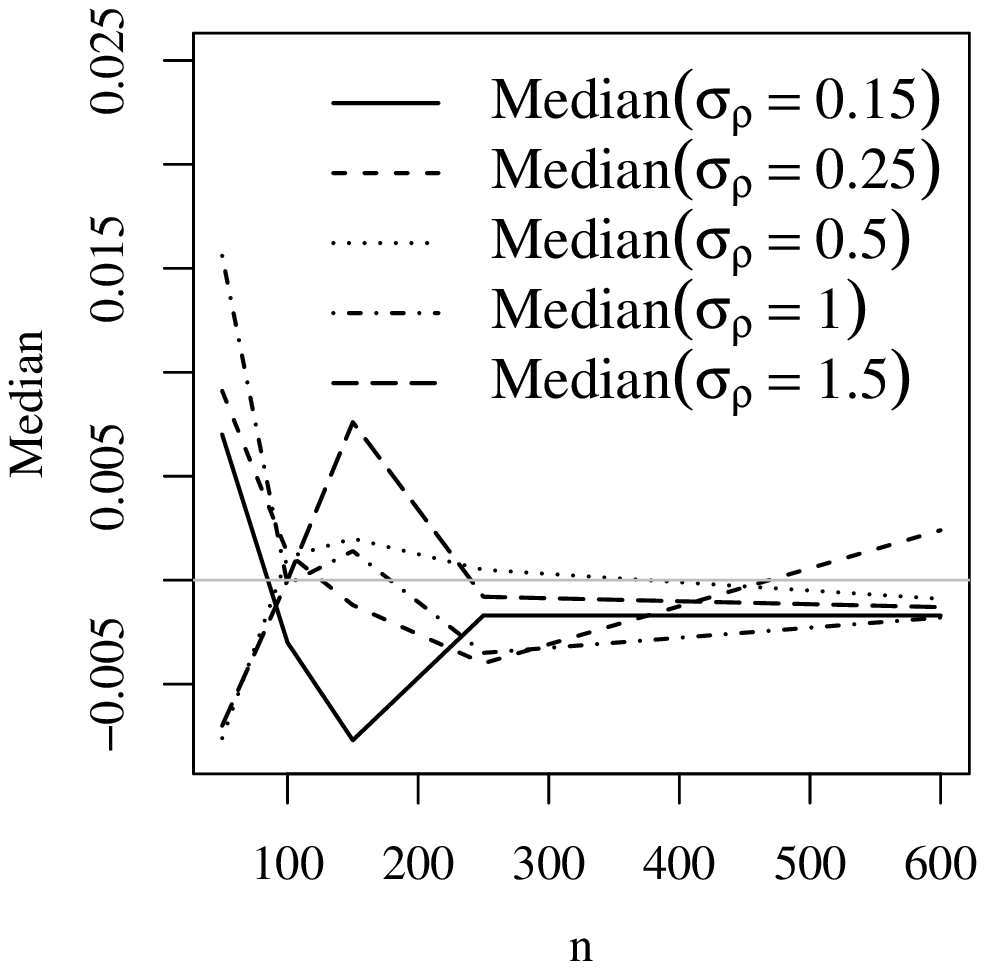}}\hspace{-0.25cm}
{\includegraphics[height=4.8cm,width=4.8cm]{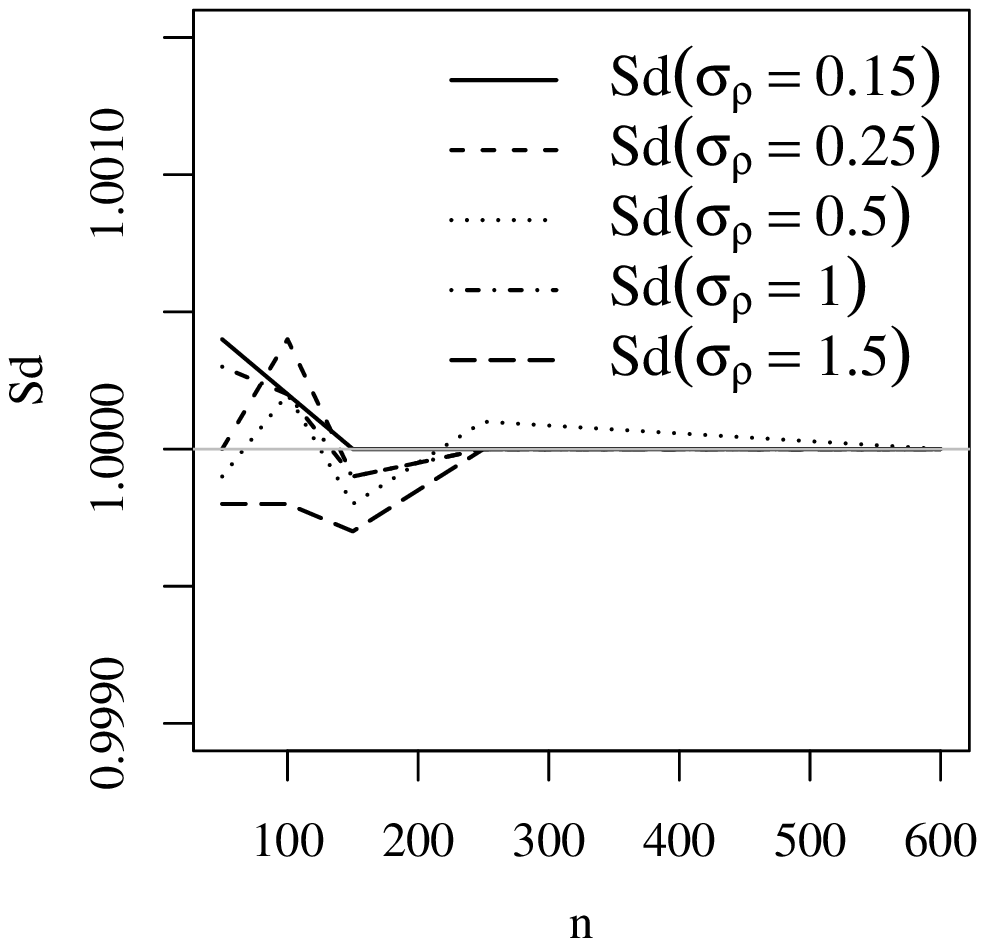}}\hspace{-0.25cm}
{\includegraphics[height=4.8cm,width=4.8cm]{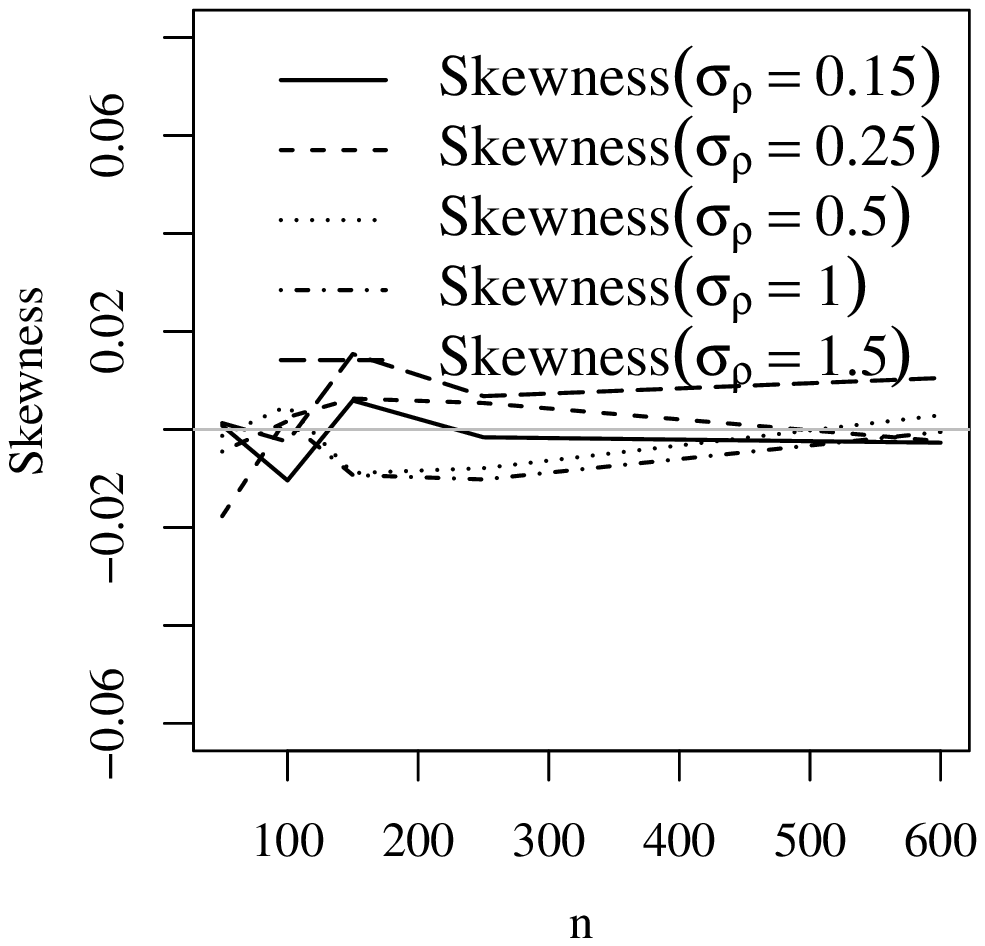}}\hspace{-0.25cm}
{\includegraphics[height=4.8cm,width=4.8cm]{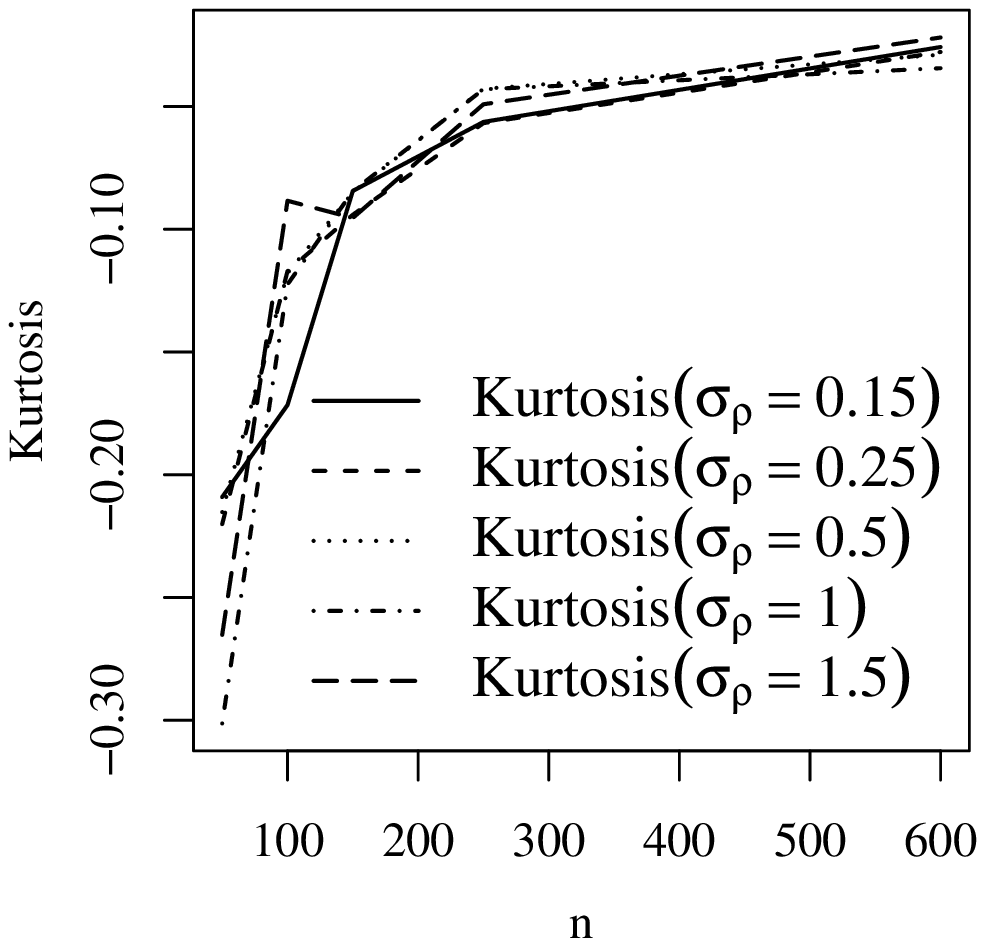}}
\vspace{-0.2cm}
\caption{Monte Carlo simulation results of the GCS residuals for the unit-log-normal model.}
\label{fig_normal_MC_RQ}
\end{figure}


\section{Illustration with internet accesss data} \label{Sec:6}
\noindent

In this section, a real-world data set, corresponding to the share of the population using internet, is analyzed. Figure \ref{fig:inter} shows the share of the population that is accessing internet for 201 countries of the world in 2015. Here, the internet can be accessed through a computer, mobile phone, games machine, digital TV, etc.; see \url{https://ourworldindata.org/internet}. From Figure \ref{fig:inter}, we observe high rates of people online in richer countries and much lower rates in the developing world. For example, three-quarters (74.55\%) of people in the US were online, while in India only 14.9\% used the internet.

\begin{figure}[!ht]
\centering
{\includegraphics[height=10cm,width=18cm]{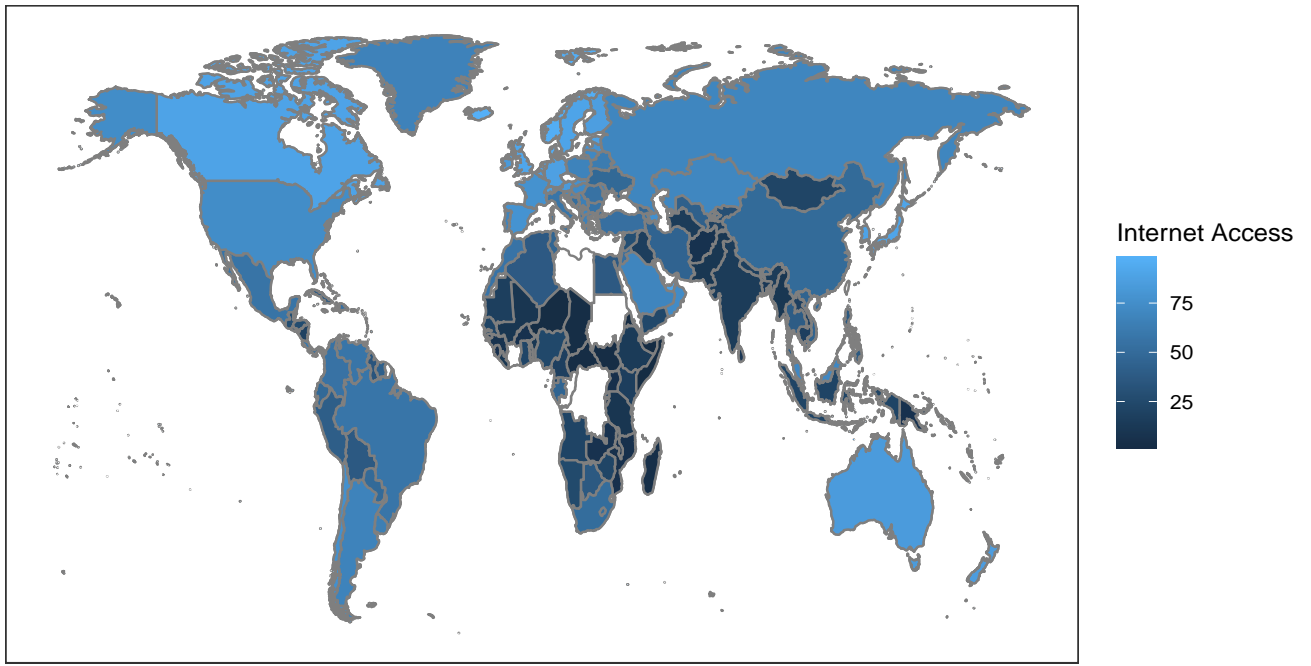}}
\vspace{-0.5cm}
 \caption{\small {Share of the population using the internet, 2015.}}
\label{fig:inter}
\end{figure}

In terms of stochastic representation, the share of the population using internet, represented by the random variable $W$, can be expressed as
\begin{equation}\label{stoch}
W=\frac{T_1}{T_1 + T_2}
\end{equation}
where $T_1 \in \mathbb{R}^+$ and $T_2 \in \mathbb{R}^+$ are two random variables following log-symmetric distributions representing the number of people with internet access and the number of people without internet access, respectively; in other words, the sum $T_1 + T_2$ represents the total population while the ratio in \eqref{stoch} has support in the unit interval $(0,1)$ with PDF in \eqref{ULS-PDF}.

Table~\ref{tab:descp_internet} provides descriptive statistics for the share of the population using the internet, including the mean, median, standard deviation (SD), coefficient of variation (CV), skewness (CS), (excess) kurtosis (CK), and minimum and maximum values. From this table, we note that the mean is almost equal to the median. In addition, the CV (dispersion around the mean) is less than 100\%, indicating a low dispersion of data around the mean. Finally, we observe that the data show skewness near zero and low degree of kurtosis.

\begin{table}[!ht]
\centering
\caption{Summary statistics for the internet data.}
\label{tab:descp_internet}
\begin{tabular}{ccccccccccccccccc}
\hline
   Mean   & Median       &  SD     &  CV       &  CS   &  CK    & minimum & maximum & size  \\
\hline
0.480    & 0.497         & 0.289  & 60.276\%  &  0.008 & -1.316 & 0.011 & 0.983  & 201 \\
\hline
\end{tabular}
\end{table}

Table \ref{tab:resulall_internet} presents the ML estimates and SEs for the unit-log-normal, unit-log-Student-$t$ and unit-log-Laplace model parameters. This table also presents the log-likelihood value, and Akaike (AIC) and Bayesian information (BIC) criteria. For comparative purpose, the results of the beta model \citep{FerrariCribari2004} are provided as well. The results in Table \ref{tab:resulall_internet} reveal that the proposed unit-log-normal and unit-log-Student-$t$ models provide better adjustments than the beta model based on the values of log-likelihood, AIC and BIC. Furthermore, the unit-log-normal model provides the best fit to these data.

\begin{table}[!ht]
\centering
\begin{small}
 \caption{{ML estimates (with SE in parentheses) and model selection measures for fit to the internet data.}}
 \begin{tabular}{l rrrrrrrrrrrrrrrrrrrrr}
\hline
          & Unit-log-normal  & Unit-log-Student-$t$  & Unit-log-Laplace     &  Beta  \\[-0.15cm]
          &          &    &           &              \\
\hline
$\sigma_{\rho}\, (\text{or}\ \mu)$ &  1.6645  &  1.5390  & 1.9794    & -0.7386     \\[-0.1cm]
                                   & (0.0830) & (0.0860) & (0.2341)  & (0.0397)    \\
$\nu \, (\text{or}\ \phi)$         &          &  10      &         & 2.1860  \\[-0.1cm]
                                   &          &          &         & (0.1849)     \\
Log-lik.                           & 3.2534   & 1.2758   & -21.9354 & 1.1450    \\
AIC                                & -4.5068  & -0.5516  & 45.8709  & 1.7098     \\
BIC                                & -1.2035  & 2.7517   & 49.1742  & 8.3164     \\
\hline
\end{tabular}
\label{tab:resulall_internet}
\end{small}
\end{table}

Figure~\ref{fig:qqplots_internet} shows the QQ plots with simulated envelope of the GCS and RQ residuals for the unit-log-normal, unit-log-Student-$t$, unit-log-Laplace and beta models considered in Table \ref{tab:resulall_internet}. We observe that the unit-log-normal model provides better fit than the unit-log-Student-$t$, unit-log-Laplace and beta models. Figure \ref{fig:pds_cdfs} displays the histogram and superimposed fitted PDFs, and fitted CDFs (empirical CDF in gray).

\begin{figure}[!ht]
\centering
\subfigure[Unit-log-normal]{\includegraphics[height=4cm,width=4cm]{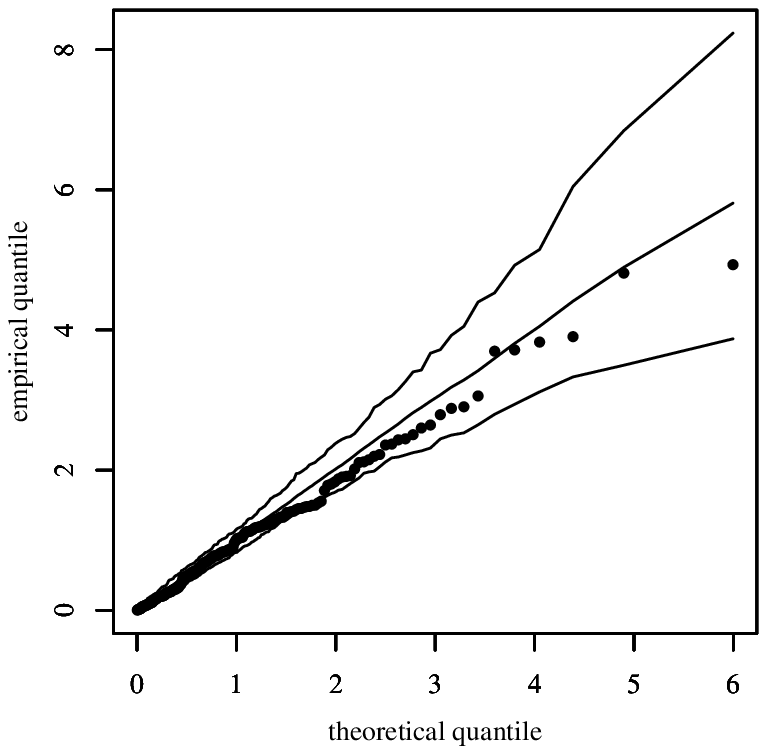}}
\subfigure[Unit-log-Student-$t$]{\includegraphics[height=4cm,width=4cm]{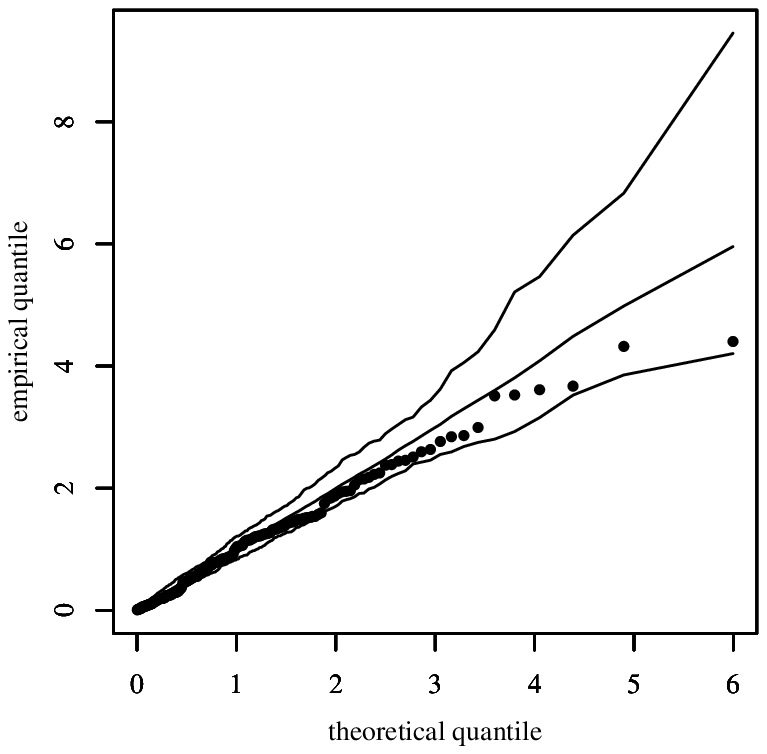}}
\subfigure[Unit-log-Laplace]{\includegraphics[height=4cm,width=4cm]{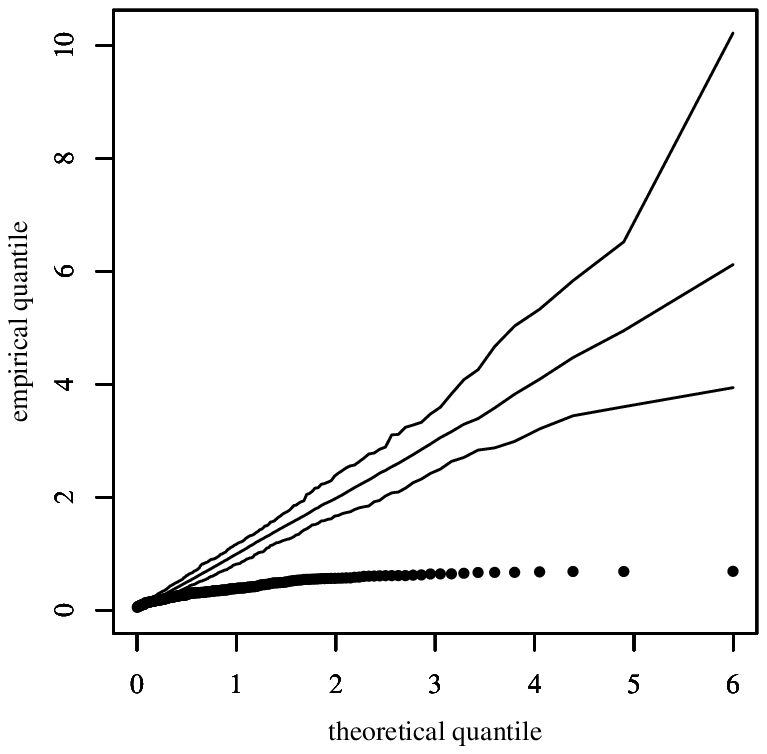}}
\subfigure[Beta]{\includegraphics[height=4cm,width=4cm]{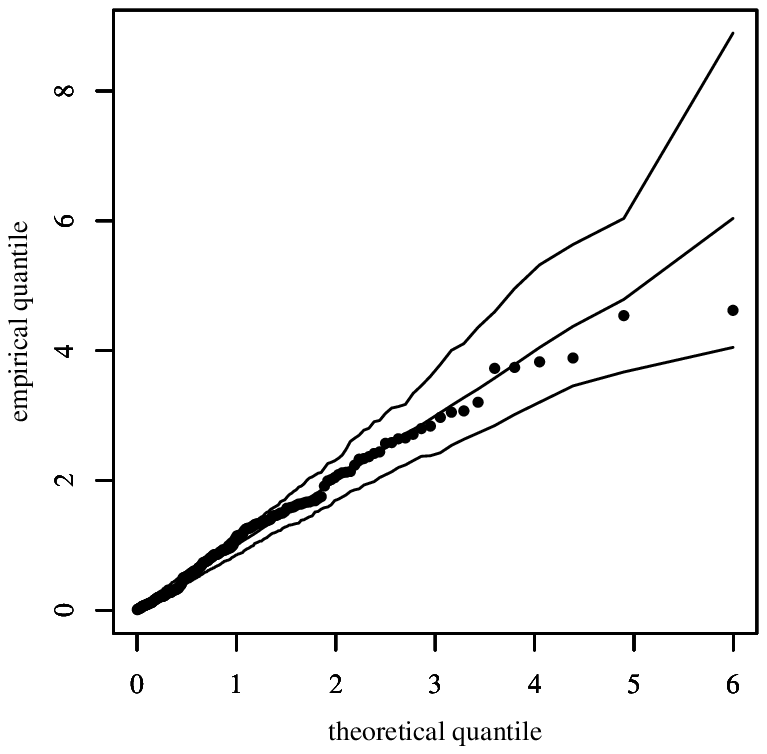}}\\
\subfigure[Unit-log-normal]{\includegraphics[height=4cm,width=4cm]{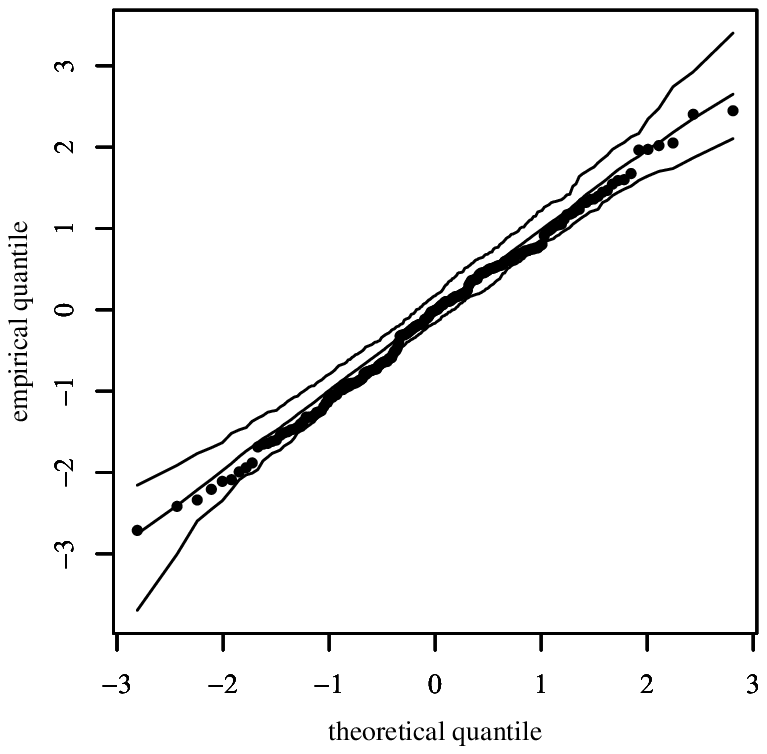}}
\subfigure[Unit-log-Student-$t$]{\includegraphics[height=4cm,width=4cm]{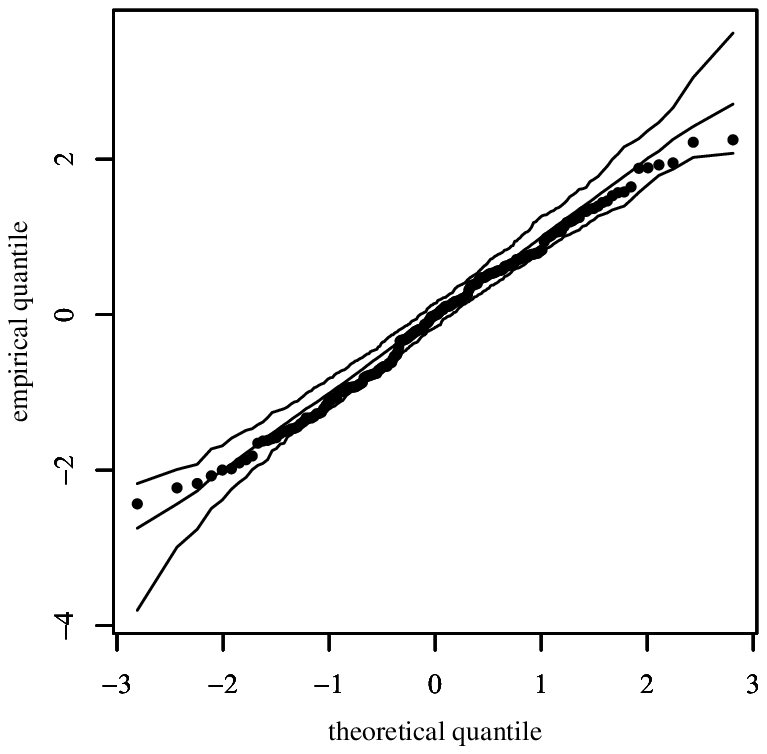}}
\subfigure[Unit-log-Laplace]{\includegraphics[height=4cm,width=4cm]{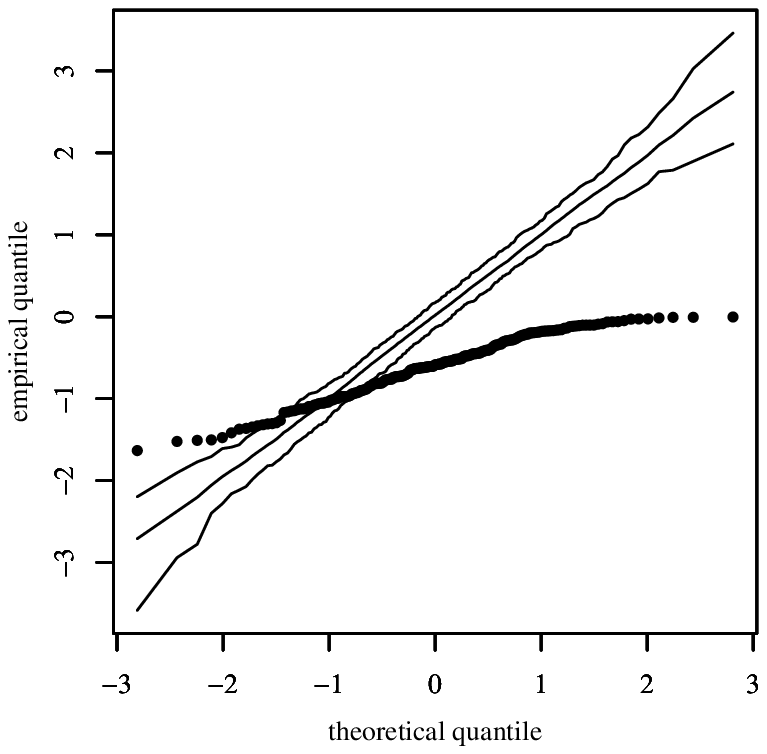}}
\subfigure[Beta]{\includegraphics[height=4cm,width=4cm]{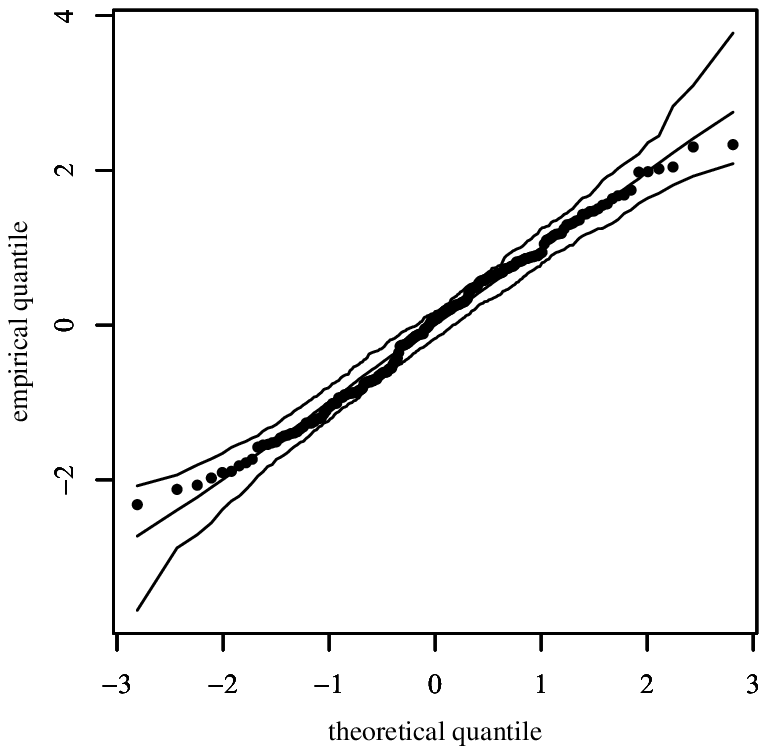}}\\
 \caption{\small {QQ plot and its envelope for the GCS (top) and RQ (bottom) residuals in the indicated model for the internet data.}}
\label{fig:qqplots_internet}
\end{figure}

\begin{figure}[!ht]
\centering
{\includegraphics[height=5cm,width=6cm]{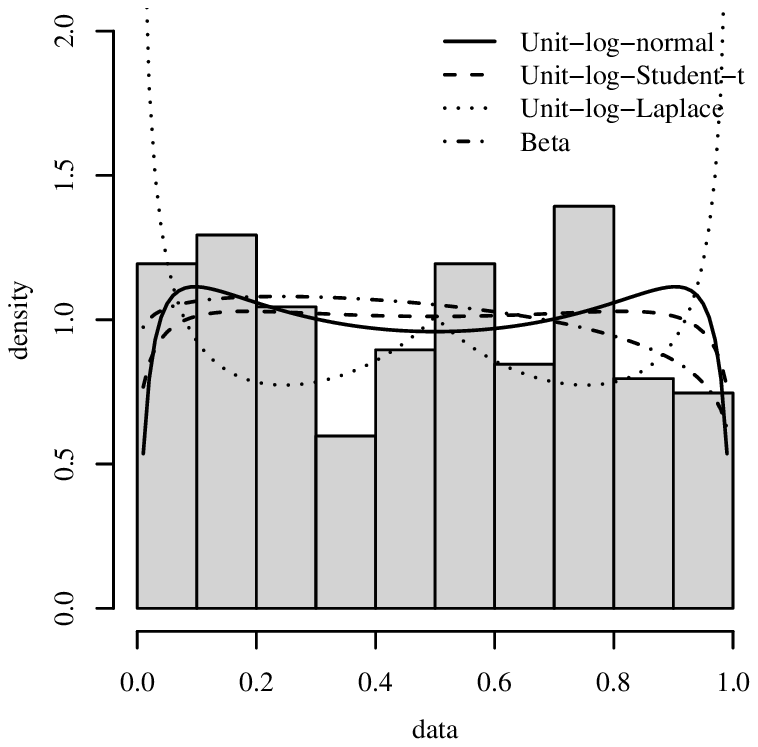}}
{\includegraphics[height=5cm,width=6cm]{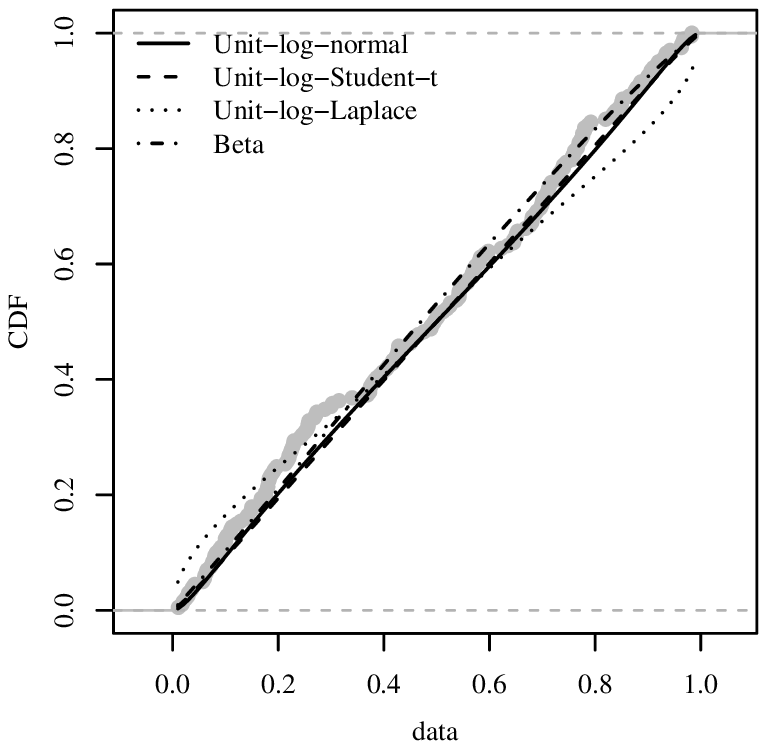}}
\vspace{-0.5cm}
 \caption{\small {Histogram of the internet data with fitted PDFs (left) and fitted CDFs (right).}}
\label{fig:pds_cdfs}
\end{figure}

\section{Concluding remarks} \label{Sec:7}
\noindent


In the present paper, we have introduced a family of flexible models over the interval $(0, 1)$. By suitably defining the density generator, we can transform any distribution over the real line into a bounded distribution over the interval $(0, 1)$. In this way, we have transformed the normal, Student-$t$ and Laplace distributions into  unit-log-normal, unit-log-Student-$t$ and unit-log-Laplace distributions, respectively. These three cases have been studied at great length.  However, many other unit-log-symmetric distributions, such as the unit-log Kotz type, unit-log Pearson Type VII, unit-log hyperbolic, unit-log slash, unit-log-logistic and unit-log-power exponential distributions can all be studied similarly. These may be useful alternatives to ``conventional'' models like beta generalizations, Kumaraswamy and the Topp-Leone distributions.

\paragraph{Acknowledgements}
This study was financed in part by the Coordenação de 
Aperfeiçoamento de Pessoal de Nível Superior - Brasil (CAPES) 
(Finance Code 001). Roberto Vila gratefully acknowledges financial support from FAP-DF, Brazil.

\paragraph{Disclosure statement}
There are no conflicts of interest to disclose.



\end{document}